\documentclass[a4paper,USenglish,runningheads]{llncs}
\usepackage{microtype}
\usepackage{xspace}
\usepackage{booktabs}
\usepackage[usenames,dvipsnames,table]{xcolor}
\usepackage{graphicx}
\DeclareGraphicsExtensions{.pdf,.png,.eps}
\graphicspath{{figures/}}
\definecolor{dark blue}{rgb}{0.121,0.47,0.705}
\newcommand{\bl}{\color{dark blue}}
\definecolor{dark brown}{rgb}{0.651, 0.337, 0.157}

\definecolor{defblue}{rgb}{0.08235294118,0.3098039216,0.537254902}
\let\emph\relax
\DeclareTextFontCommand{\emph}{\bl\em}
\newcommand{\appmark}{$\star$}

\usepackage{amssymb}
\usepackage{mathtools}
\DeclarePairedDelimiter\set{\{}{\}}
\DeclarePairedDelimiter\abs{\lvert}{\rvert}
\DeclarePairedDelimiter\ceil{\lceil}{\rceil}
\DeclarePairedDelimiter\floor{\lfloor}{\rfloor}
 
\def\to{\ensuremath{\rightarrow}} 

\def\cP{\ensuremath{\mathcal{P}}\xspace}
\DeclareMathOperator{\slope}{slope}
\DeclareMathOperator{\seg}{seg}
\DeclareMathOperator{\openseg}{port}
\DeclareMathOperator{\arc}{arc}

\DeclareMathOperator{\V}{V}
\DeclareMathOperator{\E}{E}
\DeclareMathOperator{\N}{N}
\DeclareMathOperator{\Lpath}{L}
\DeclareMathOperator{\Rpath}{R}
\DeclareMathOperator{\cw}{cw}
\DeclareMathOperator{\ccw}{ccw}

\usepackage{etoolbox}

\newtheorem{myclaim}{Claim}
\usepackage{hyperref}
\usepackage[shortlabels]{enumitem}
\usepackage{thm-restate}

\newtoggle{saveSpace}
\togglefalse{saveSpace}

\usepackage{url}
\usepackage[capitalise,nameinlink]{cleveref}
\usepackage{cite}
\usepackage[font=small, labelfont=bf]{caption}
\usepackage[font=small, labelfont=small]{subcaption}

\renewcommand{\orcidID}[1]{\href{https://orcid.org/#1}{\includegraphics[scale=.03]{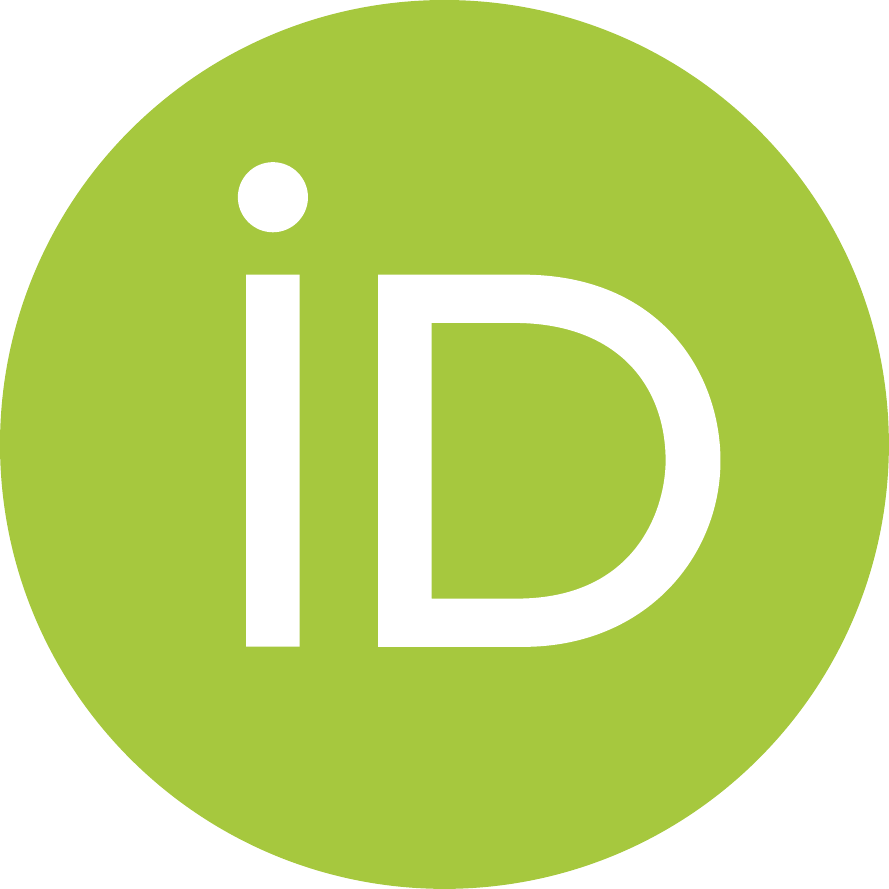}}}

\let\doendproof\endproof
\renewcommand\endproof{~\hfill$\qed$\doendproof}

\usepackage{colortbl}
\definecolor{lightGray}{RGB}{222, 222, 222}

\title{The Segment Number: \\ Algorithms and Universal Lower Bounds \\
  for Some Classes of Planar Graphs}

\titlerunning{The Segment Number: Algorithms and Universal Lower Bounds}

\author{Ina~Goe\ss{}mann\inst1 \and
  Jonathan~Klawitter\inst1\orcidID{0000-0001-8917-5269} \and
  Boris~Klemz\inst1\orcidID{0000-0002-4532-3765} \and
  Felix~Klesen\inst1\orcidID{0000-0003-1136-5673} \and
  Stephen~Kobourov\inst2\orcidID{0000-0002-0477-2724} \and
  Myroslav~Kryven\inst2\orcidID{0000-0003-4778-3703} \and
  Alexander~Wolff\inst1\orcidID{0000-0001-5872-718X} \and
  Johannes~Zink\inst1\orcidID{0000-0002-7398-718X}}

\authorrunning{Goe\ss{}mann et al.}

\institute{Universität Würzburg, Würzburg, Germany \and University of
  Arizona, Tucson, USA}

\usepackage[disable]{todonotes}

\newcommand{\bk}[1]{\todo[color=Tan!50]{bk: #1}}

\newcommand{\lncsarxiv}[2]{#2}

\begin{document}

\crefname{figure}{Fig.}{Figs.}
\crefname{theorem}{Thm.}{Thms.}
\crefname{corollary}{Cor.}{Cors.}
\crefname{lemma}{Lem.}{Lems.}
\crefname{proposition}{Prop.}{Props.}
\crefname{section}{Sect.}{Sects.}
\crefname{equation}{Eq.}{Eqs.}
\crefname{myclaim}{Claim}{Claims}
\crefname{appendix}{App.}{App.}
\Crefname{figure}{Figure}{Figures}
\Crefname{theorem}{Theorem}{Theorems}
\Crefname{corollary}{Corollary}{Corollaries}
\Crefname{lemma}{Lemma}{Lemmas}
\Crefname{proposition}{Proposition}{Propositions}
\Crefname{section}{Section}{Sections}
\Crefname{equation}{Equation}{Equations}

\maketitle

\begin{abstract}
  The \emph{segment number} of a planar graph~$G$ is the smallest
  number of line segments needed for a planar straight-line drawing
  of~$G$.  Dujmovi{\'c}, Eppstein, Suderman, and Wood
  [CGTA'07] introduced this measure for the {\em visual
    complexity} of graphs.  There are optimal algorithms for trees
  and worst-case optimal algorithms for outerplanar graphs, 2-trees,
  and planar 3-trees.  It is known that every {\em cubic}
  triconnected planar $n$-vertex graph (except~$K_4$) has segment number $n/2+3$, which is the only known {\em universal} lower bound for a meaningful class of
  planar graphs.

  \hspace{2.5ex}
  We show that every triconnected planar 4-regular graph can be drawn
  using at most $n+3$ segments.  This bound is tight up to an additive
  constant, improves a previous upper bound of $7n/4+2$ implied by a
  more general result of Dujmovi{\'c} et al., and supplements the
  result for cubic graphs.  We also give a simple optimal algorithm
  for cactus graphs, generalizing the above-mentioned result for
  trees.  We prove the first linear universal lower bounds for
  outerpaths, maximal outerplanar graphs, 2-trees, and planar 3-trees.
  This shows that the existing algorithms for these graph classes are
  constant-factor approximations.  For maximal outerpaths, our bound
  is best possible and can be generalized to circular~arcs.

\end{abstract}

\section{Introduction}

A drawing of a given graph can be evaluated by various quality
measures depending on the concrete purpose of the drawing.  Classic
examples of such measures include drawing area, number of edge
crossings, neighborhood preservation, and stress of the embedding.
More recently, Schulz~\cite{s-dgfa-JGAA15} proposed the \emph{visual
  complexity} of a drawing, determined by the number of geometric
objects (such as line segments or circular arcs) that the drawing
consists~of.  It has been experimentally verified that
people without mathematical background tend to prefer drawings with
low visual complexity~\cite{kms-eaadfs-JGAA18}.  The visual complexity
of a graph drawing depends on the drawing style,
as well as on the underlying graph properties.  A well-studied measure
of the visual complexity of a graph is its segment number, introduced
by Dujmovi\'c, Eppstein, Suderman, and Wood~\cite{desw-dpgfss-CGTA07}.
It is defined as follows.  \iftoggle{saveSpace}{A}{Recall that a}
\emph{straight-line drawing}
of a graph maps (i)~the vertices of the graph injectively to points in
the plane and (ii)~the edges of the graph to straight-line segments
that connect the corresponding points.  A \emph{segment} in such a
drawing is a maximal set of edges that together form a line segment.
Given a straight-line drawing~$\Gamma$ of a graph, the set of
segments it induces is unique.  \iftoggle{saveSpace}{Its
  cardinality}{The cardinality of that set} is the
\emph{segment number} of~$\Gamma$.  The \emph{segment number},
$\seg(G)$, of a planar graph~$G$ is the smallest segment number over
all crossing-free straight-line drawings of~$G$.

\paragraph{Previous work.}

Dujmovi\'c et al.~\cite{desw-dpgfss-CGTA07} pointed out two natural
lower bounds for the segment number: (i)~$\eta(G)/2$, where $\eta(G)$
is the number of odd-degree vertices of~$G$, and (ii) the \emph{slope number},
$\slope(G)$, of~$G$, which is defined as follows.  The slope number
$\slope(\Gamma)$ of a straight-line drawing~$\Gamma$ of~$G$ is the
number of different slopes used by any of the straight-line edges
in~$\Gamma$.  Then $\slope(G)$ is the minimum of $\slope(\Gamma)$ over
all straight-line drawings~$\Gamma$ of~$G$.  Dujmovi\'c et al.\ also
showed that any tree~$T$ admits a drawing with $\seg(T)=\eta(T)/2$
segments and $\slope(T)=\Delta(T)/2$ slopes, where $\Delta(T)$ is the
maximum degree of a vertex in~$T$.  These drawings, however, use
exponential area.
Recall that an \emph{outerplanar graph} is a plane graph that can be
drawn such that all vertices lie on the outer face.
The \emph{weak dual graph} of an outerplane graph is its dual graph
without the vertex corresponding to the outer face; it is known to be
a tree.  An outerplane graph whose weak dual is a path is called an
\emph{outerpath}.  A \emph{maximal outerplanar graph} is an
outerplanar graph with the maximum number of edges.
Dujmovi{\'c} et al.\ showed that every maximal outerplanar
graph~$G$ with $n$ vertices admits an {\em outerplanar} straight-line
drawing with at most $n$ segments.  They showed that this is
worst-case optimal.
They also gave (asymptotically) worst-case optimal algorithms for
2-trees and plane (where the combinatorial embedding and outer face is
fixed) 3-trees.  Finally, they showed that every triconnected planar
graph with $n$ vertices can be drawn using at most $5n/2-3$ segments.
For the special cases of triangulations and 4-connected
triangulations, Durocher and Mondal~\cite{dm-dptfs-CGTA19} improved
the upper bound of Dujmovi\'c et al.\ to $(7n-10)/3$ and $(9n-9)/4$,
respectively.  The former bound implies a bound of $(16n - 3m - 28)/3$
for arbitrary planar graphs with $n$ vertices and $m$ edges.
Kindermann et al.~\cite{kmss-dpgfs-GD19} observed that this implies
that $\seg(G) \le (8n-14)/3$ for any planar graph~$G$: if
$m>(8n-14)/3$ this follows from the bound, otherwise
any
drawing of~$G$ is good enough.
Constructive linear-time algorithms that compute
the segment number of series-parallel graphs of maximum degree~3
and of maximal outerpaths were given by
Samee et al.~\cite{saar-msdsp-GD08}
and by Adnan~\cite{Adnan2008}, respectively.
Mondal et al.~\cite{mnbr-mscd3ccpg-JCO13} and
Igamberdiev et al.~\cite{ims-dpc3cgfsae-JGAA17} showed that every
cubic triconnected planar graph (except~$K_4$) has segment number
$n/2+3$.
Hültenschmidt et al.~\cite{hkms-dttfg-JGAA18} showed that
trees, maximal outerplanar graphs and planar 3-trees
admit drawings on a
grid of polynomial size, using slightly more segments.
Kindermann et al.~\cite{kmss-dpgfs-GD19} improved some of these bounds.
\iftoggle{saveSpace}{ }{ Concerning the computational complexity,
  Durocher et al.~\cite{dmnw-nmsdp-JGAA13} showed that the segment
  number of a planar graph is NP-hard to compute, even if one insists
  that in the resulting planar drawing all faces are convex.
}%

\paragraph{Other related work.}
Okamoto et al.~\cite{orw-ysng-GD19} investigated variants of the
segment number.  For planar graphs in~2D, they allowed bends.  For
arbitrary graphs, they considered crossing-free straight-line drawings
in~3D and straight-line drawings with crossings in~2D.  They showed
that all segment number variants are $\exists\mathbb{R}$-complete to
compute, and they gave upper and existential lower bounds for the
segment number variants of cubic graphs.
The \emph{arc number}, $\arc(G)$, of a graph~$G$ is the smallest
number of circular arcs in any circular-arc drawings of~$G$.  It has
been introduced by Schulz\cite{s-dgfa-JGAA15}, who gave algorithms for
drawing series-parallel graphs, planar 3-trees, and triconnected
planar graphs with few circular arcs.  For trees, he
reduced the drawing area (from exponential to polynomial).
Chaplick et al.~\cite{cflrvw-dgflf-JoCG20,cflrvw-cdgfl-WADS17}
considered a different measure of the visual complexity, namely the
number of lines (or planes) needed to cover crossing-free
straight-line drawings of graphs in 2D (and 3D).  Kryven et
al.~\cite{krw-dgfcf-JGAA19} considered spherical covers.

\begin{table}[t]
  \footnotesize\centering
  \begin{tabular}{@{}lcccccccc@{}}
    \toprule
    & \multicolumn{2}{c}{Universal} & \multicolumn{2}{c}{Existential}
    & \multicolumn{2}{c}{Existential} & \multicolumn{2}{c}{Universal} \\ 
    Graph class & \multicolumn{2}{c}{lower bound}
                                    & \multicolumn{2}{c}{upper bound} & \multicolumn{2}{c}{lower bound}
                                      & \multicolumn{2}{c}{upper bound} \\
    \midrule
    planar conn. & 1 & & 1 & & $2n-2$ & \cite{desw-dpgfss-CGTA07}
                                                                      & $(8n-14)/3$ & \cite{dm-dptfs-CGTA19, kmss-dpgfs-GD19} \\
    planar 3-conn.\ & $\sqrt{2n}$ & {\cite{desw-dpgfss-CGTA07}}
    & $O(\sqrt{n})$ & \cite{desw-dpgfss-CGTA07} & $2n-6$
                                    & \cite{desw-dpgfss-CGTA07} & $5n/2-3$ & \cite{desw-dpgfss-CGTA07}\\
    planar 3-conn.\ 4-reg. &\cellcolor{lightGray}$\Omega(\sqrt{n})$ &\cellcolor{lightGray}R\ref{rem:triconn-universal}
    &\cellcolor{lightGray}$O(\sqrt{n})$ &\cellcolor{lightGray}R\ref{rem:triconn-universal} &\cellcolor{lightGray}$n$ &\cellcolor{lightGray}P\ref{prop:square}
                                                                      &\cellcolor{lightGray}$n+3$ &\cellcolor{lightGray}T\ref{thm:fourUnivUpper} \\
    planar 3-conn.\ 3-reg. & $n/2+3$ & \cite{desw-dpgfss-CGTA07}
    & --- & & --- & & $n/2+3$ &
                                \cite{mnbr-mscd3ccpg-JCO13,ims-dpc3cgfsae-JGAA17} \\
    triangulation & $\Omega(\sqrt{n})$ & \cite{desw-dpgfss-CGTA07}
    & $O(\sqrt{n})$ & \cite{desw-dpgfss-CGTA07}
                & $2n-2$ & \cite{desw-dpgfss-CGTA07}
                                                                      & $(7n-10)/3$ & \cite{dm-dptfs-CGTA19} \\
    4-conn.\ triangulation
    & $\Omega(\sqrt{n})$ & \cite{desw-dpgfss-CGTA07}
    & $O(\sqrt{n})$ & \cite{desw-dpgfss-CGTA07}
                & $2n-6$ & \cite{desw-dpgfss-CGTA07}
                                                                      & $(9n-9)/4$ & \cite{dm-dptfs-CGTA19} \\
    planar 3-trees &\cellcolor{lightGray}$n+4$ &\cellcolor{lightGray}T\ref{clm:3treeseg}
    &\cellcolor{lightGray}$n+7$ &\cellcolor{lightGray}P\ref{clm:good3tree3} &\cellcolor{lightGray}$3n/2$ & \cellcolor{lightGray}P\ref{clm:bad3tree3}
                                                                      & $2n-2$ & \cite{desw-dpgfss-CGTA07}\\
    2-trees &\cellcolor{lightGray}$(n+7)/5$ &\cellcolor{lightGray}T\ref{clm:outerplanarSec}
    &\cellcolor{lightGray}$(5n+24)/13$
                                      &\cellcolor{lightGray}P\ref{clm:outerplanarSecUpper}
                & $3n/2-2$ & \cite{desw-dpgfss-CGTA07}
                                                                      & $3n/2$ & \cite{desw-dpgfss-CGTA07}\\
    maximal outerplanar &\cellcolor{lightGray}$(n+7)/5$ &\cellcolor{lightGray}T\ref{clm:outerplanarSec}
    &\cellcolor{lightGray}$(5n+24)/13$&\cellcolor{lightGray}P\ref{clm:outerplanarSecUpper}
                & $n$ & \cite{desw-dpgfss-CGTA07} & $n$ & \cite{desw-dpgfss-CGTA07} \\
    maximal outerpath &\cellcolor{lightGray}$\lfloor n/2 \rfloor + 2$ &\cellcolor{lightGray}T\ref{clm:outerpath-segs}
    &\cellcolor{lightGray}$\lfloor n/2 \rfloor + 2$ &\cellcolor{lightGray}P\ref{prop:outerpathExamples}
                & $n$ & \cite{desw-dpgfss-CGTA07} & $n$ & \cite{desw-dpgfss-CGTA07} \\
    cactus &\cellcolor{lightGray}$\eta/2+\gamma$ &\cellcolor{lightGray}L\ref{lem:cactus-structure}
    &\cellcolor{lightGray}--- &\cellcolor{lightGray}&\cellcolor{lightGray}--- &\cellcolor{lightGray}&\cellcolor{lightGray}$\eta/2+\gamma$ &\cellcolor{lightGray}T\ref{thm:cactus-algo} \\
    \bottomrule
  \end{tabular}
  
  \medskip
  \caption{
    \iftoggle{saveSpace}{Bounds  
}{
    Universal and existential lower and upper
    bounds
}%
    on the segment number for subclasses of planar
    graphs.  By {\em existential upper bound} we mean an upper bound
    for the universal lower bound.
    \iftoggle{saveSpace}{}{Such a bound is provided by the
    segment number of a
    specific graph family within the given graph class.}%
    Here, $\eta$ is the number of odd-degree vertices and
    $\gamma=3c_0+2c_1+c_2$, where $c_i$ is the
      number of simple cycles with exactly $i$ cut vertices.
     \iftoggle{saveSpace}{}{We use ``---'' to indicate that universal
      lower bound and the universal upper bound agree for a specific
      graph class.  The corresponding algorithms are thus optimal.
    Results of this paper are shaded in gray,
    where we link to remarks (R), lemmas~(L), propositions (P), theorems (T).
    }%
  }
  \label{tab:results}\vspace*{-4.5ex}
\end{table}

\vspace*{-1.5ex}

\paragraph{Contribution and outline.}

In terms of universal upper bounds, we first show that every triconnected planar 4-regular graph
with $n$ vertices 
can be drawn using at most
$n+3$ segments (note that there are $2n$ edges); see
\cref{sub:4-regular}.  This bound is tight up to an additive constant,
improves a previous upper bound of $7n/4+2$ implied by a more general
result~\cite[Thm.~15]{desw-dpgfss-CGTA07} of Dujmovi{\'c} et al.,
and supplements the result for cubic graphs due to Mondal et
al.~\cite{mnbr-mscd3ccpg-JCO13} and Igamberdiev et
al.~\cite{ims-dpc3cgfsae-JGAA17}.
Our algorithm works even
for {\em plane} graphs and produces drawings that are \emph{convex},
that is, the boundary of each face corresponds to a convex polygon.
We remark that triconnected planar 4-regular graphs are a rich and natural graph class that comes with a simple set of generator rules~\cite{DBLP:journals/jgt/BroersmaDG93}.
It might seem tempting to prove our result inductively by means of these rules, though we have not been able make this idea work.
Instead, our algorithm relies on a decomposition of the graph along carefully chosen paths (\cref{lem:windmill}), which might be of independent interest.
We also give a simple optimal (cf.\ Table~\ref{tab:results}) algorithm
for cactus graphs\footnote{
A \emph{cactus} is a connected graph where any two simple cycles
share at most one vertex.}
(see \cref{sec:ratio}),
generalizing the result of Dujmovi{\'c} et al.\ for trees.

We prove the first linear universal lower bounds for maximal outerpaths
($\floor{n/2} +2$; see \cref{sec:outerpaths}), maximal outerplanar
graphs as well as 2-trees ($(n+7)/5$; see \cref{sec:ratio}),
and planar 3-trees ($n+4$; see \cref{sec:ratio}).
This makes the corresponding algorithms of Dujmovi\'c et al.\
constant-factor approximation algorithms.
For Adnan's algorithm~\cite{Adnan2008} that computes the segment
number of maximal outerpaths, our result provides a lower bound on the
size of the solution.  For maximal outerpaths, our
bound is best possible and can be generalized to circular arcs.  For
planar 3-trees, the bound is best possible up to the additive
constant.
Known and new results are listed in Table~\ref{tab:results}.
Claims \lncsarxiv{marked }{}with ``\appmark'' are proved
in~\lncsarxiv{\cite{arxiv2022}}{the appendix}.

\paragraph{Notation and terminology.}
\label{sec:prelim}

All graphs in this paper are simple\iftoggle{saveSpace}{}{(i.e., we do
  not allow parallel edges or self-loops)}.
For any graph~$G$, let $\V(G)$ be the vertex set and $\E(G)$ the edge
set of~$G$.  Now let~$\Gamma$ be a planar drawing of a
planar and connected graph~$G$.
The boundary~$\partial f$ of each face~$f$ of~$\Gamma$ can be
uniquely described by a counterclockwise sequence of edges.  If~$G$ is
biconnected, then~$\partial f$ is a simple cycle%
\iftoggle{saveSpace}{}{( otherwise, $\partial f$ can visit vertices and
  edges multiple times)}.
The collection of the boundaries of all faces of $\Gamma$ is called the
\emph{combinatorial embedding} of~$\Gamma$.  The unique unbounded face
of~$\Gamma$ is called its \emph{outer} face; the remaining faces are
called \emph{internal}.  Vertices (edges) belonging to the boundary of
the outer face are called \emph{outer} vertices (edges); the remaining
vertices (edges) are called \emph{internal}.  A \emph{plane} graph is
a planar graph equipped with a combinatorial embedding and a
distinguished outer face.
\iftoggle{saveSpace}{}{Note that two drawings of~$G$ with the same
  combinatorial embedding may have different outer faces.
}%
A path in a plane graph is \emph{internal} if its edges and interior
vertices do not belong to its outer face.
\iftoggle{saveSpace}{ }{ We say that an angle is \emph{convex} if it
  is at most $\pi$ and \emph{reflex} if it exceeds $\pi$.  In a
  \emph{convex} polygon, each internal angle is convex.
}%
For any $k \in \mathbb{N}$, we use $[k]$ as shorthand for
$\{1,2,\dots,k\}$.

\section{Triconnected 4-Regular Planar Graphs}
\label{sub:4-regular}

This section is concerned with the segment number of $3$-connected
$4$-regular planar graphs.  We establish a universal upper bound of
$n+3$ segments, which we complement with an existential lower bound of
$n$ segments, where $n$ denotes the number of vertices.

\paragraph{Overview.}

Towards the upper bound, we will show that each graph of the
considered class admits a drawing where all but three of its vertices
are placed in the interior of some segment.  In such a drawing, each
of these vertices is the endpoint of at most two segments.  The
claimed bound then follows from the fact that each segment has exactly
two endpoints.

To construct the desired drawings, we follow a strategy that has
already been used in an algorithm by Hong and
Nagamochi~\cite{DBLP:journals/jda/HongN10},
which was sped up by Klemz~\cite{DBLP:conf/esa/Klemz21}.
Both algorithms generate
convex drawings of so-called hierarchical plane st-graphs, but they
can also be applied to ``ordinary'' plane graphs.  In this context,
the algorithmic framework is as follows: the input is an internally
(defined below, see \cref{def:int3conDefs}) $3$-connected plane
graph~$G$ and a convex drawing~$\Gamma^o$ of the boundary of its outer
face.  The task is to extend~$\Gamma^o$ to a convex drawing of~$G$.
The main idea of both algorithms is to choose a suitable internal
vertex~$y$ of the given graph~$G$ and compute three disjoint (except
for~$y$) paths $P_1,P_2,P_3$ from~$y$ to the outer face.  Each of
these paths is then embedded as a straight-line segment so
that~$\Gamma^o$ is dissected into three convex polygons, for an
illustration see \cref{fig:leftAligned}a.  The graphs corresponding to
the interior of these polygons can now be handled recursively.  To
ensure that a solution exists, the computed paths (as well as the
paths corresponding to the segments of~$\Gamma^o$) need to be
\emph{archfree}, meaning that they are not arched by an internal face:
a path~$P$ is \emph{arched} by a face~$a$ between
$u, v \in \V(\partial a)\cap\V(P)$ if the subpath~$P_{uv}$ of~$P$
between~$u$ and~$v$ is interior-disjoint from~$\partial a$, see
\cref{fig:leftAligned}a.  Indeed, if~$a$ is internal, then such a
path~$P$ cannot be realized as a straight-line segment in a convex
drawing since the interior of the segment~$uv$ has to be disjoint from
the realization of~$a$.  We follow the idea of dissecting our graphs
along archfree paths.  However, to ensure that each internal vertex is
placed in the interior of some segment, the way in which we construct
our paths is necessarily%
\iftoggle{saveSpace}{ }{\footnote{Note that it is not necessarily
    possible to embed two of the paths $P_1,P_2,P_3$ on a common
    segment since their outer endpoints might already belong to a
    common segment of $\Gamma^o$ (in particular, this is the case when
    $|\V(\Gamma^o)|=3$).  Moreover, the concatenation of the two paths
    might not be archfree.}  \bk{footnote disappears when `saveSpace'
    is toggled}}%
quite different.  Specifically, we will show that a large subfamily of
the considered graph class can be dissected along three archfree paths
that are arranged in a windmill pattern as depicted in
\cref{fig:fourRegularAlgo2sketch}a.

\iftoggle{saveSpace}{}{We begin by discussing necessary conditions for
  the existence of convex drawings and the construction of archfree
  paths.  We then define the desired windmill configuration and give a
  necessary and sufficient criterion for its existence.  Finally, we
  describe our drawing algorithm, thereby establishing the universal
  upper bound, and conclude with the existential lower bound.\bk{this
    final paragraph disappears when `saveSpace' is toggled}
}%

\paragraph{Existence of convex drawings.}

It is well-known that a plane graph admits a convex drawing if and
only if it is a subdivision of an \emph{internally $3$-connected}
graph~\cite{tutte-1960,Thomassen-1984,DBLP:journals/jda/HongN10,DBLP:journals/dam/HongN08}.
There are multiple ways to define this property and it will be
convenient to refer to all of them.  Therefore, we use the following well-known  
characterization; for a proof, see, e.g., \cite{DBLP:journals/comgeo/KleistKLSSS19}.

\begin{definition}\label{def:int3conDefs}
  Let~$G$ be a plane $2$-connected graph.  Let~$o$
   denote its
  outer face. Then~$G$ is called \emph{internally $3$-connected} if
  and only if the following equivalent statements are satisfied:
  \begin{enumerate}[label=({I}\arabic*),left=0pt,nosep]
  \item \label{I1} Inserting a new vertex~$v$ in~$o$ and adding
    edges between~$v$ and all vertices of~$\partial o$ results in a
    $3$-connected graph.
  \item \label{I2} From each internal vertex~$w$ of $G$ there exist
    three paths to~$o$ that are pairwise disjoint except for the
    common vertex~$w$.
  \item \label{I3} Every separation pair $u,v$ of~$G$ is
    \emph{external}, i.e., $u$ and $v$ lie on~$\partial o$ and every
    connected component of the subgraph of $G$ induced
    by~$\V(G)\setminus\{u,v\}$ contains a vertex of~$\partial o$.
  \end{enumerate}
\end{definition}

\iftoggle{saveSpace}{}{When dissecting a $3$-connected plane graph
  along internal paths, the resulting subgraphs are not necessarily
  $3$-connected anymore.  In contrast, {\em internal} $3$-connectivity
  is preserved:\bk{disappears when `saveSpace' is toggled}
}%
\begin{restatable}[{\hyperref[obs:internally3conCycle*]{\appmark}},
  folklore]{observation}{intThreeConCycle}
  \label{obs:internally3conCycle}
  Let $G$ be an internally $3$-connected plane graph, and let $C$ be a
  simple cycle in $G$.  The closed interior~$C^-$ of~$C$ is an
  internally $3$-connected plane graph.
\end{restatable}

In the context of our recursive strategy, we face a special case of
the following problem: given an internally $3$-connected plane graph
$G$ and a convex drawing~$\Gamma^o$ of the boundary of its outer face,
extend~$\Gamma^o$ to a convex drawing of~$G$.  It is known that such
an extension exists if and only if each segment of~$\Gamma^o$
corresponds to an archfree path
of~$G$~\cite{tutte-1960,Thomassen-1984,DBLP:journals/jda/HongN10,DBLP:journals/dam/HongN08}.
Hence, we say that~$\Gamma^o$ is \emph{compatible} with~$G$ if and
only if it satisfies this property.

\paragraph{Construction of archfree paths.}

The following lemma gives rise to a strategy for transforming a given
internal path into an archfree path:

\begin{lemma}[{\hspace{1sp}\cite[Lem.~1]{DBLP:journals/jda/HongN10}}]\label{lem:minusTwo}
  Let~$G$ be an internally $3$-connected plane graph,
  \iftoggle{saveSpace}{$f$}{and let~$f$ be} an
  internal face of~$G$.  Any subpath~$P$ of~$\partial f$ with
  $|\E(P)|\le |\E(\partial f)|-2$ is archfree.
\end{lemma}

\iftoggle{saveSpace}{}{One can simply replace the arched parts by
  appropriate pieces of the boundaries of the arching faces.  More
  precisely, this strategy works as follows: }%
Let $G$ be an internally $3$-connected graph.  Consider the edges of
the outer face~$\partial o$ of~$G$ to be directed in
counterclockwise direction.  Assume that there are two distinct
vertices~$s'$ and~$t'$ on~$\partial o$ that are joined by a simple
internal path~$P'$.  Consider~$P'$ to be directed from~$s'$ to~$t'$
and let~$P=(s,\dots,t)$ be a directed subpath of~$P'$.  Suppose
that~$P$ is arched by an internal face~$a$.  Then we say~$a$
arches~$P$ \emph{from the left} if~$a$ is interior to the cycle formed
by~$P'$ and the directed $t's'$-path on~$\partial o$; otherwise, we
say that $a$ arches $P$ \emph{from the right}.  The
\emph{left-aligned} path $\Lpath_G(P)$ of~$P$ is obtained be
exhaustively applying the following modification (for an illustration
see \cref{fig:leftAligned}b): suppose that an internal face~$a$
arches~$P$ from the left between two vertices~$u,v$ such $u$
precedes~$v$ along~$P$.  Transform~$P$ by replacing its~$uv$-subpath
with the~$uv$-path obtained by walking along~$\partial a$ in
counterclockwise direction from~$u$ to~$v$.  The \emph{right-aligned}
path $\Rpath_G(P)$ is defined symmetrically.
 
\begin{figure}[tb]
  \centering 
  \includegraphics[page=10]{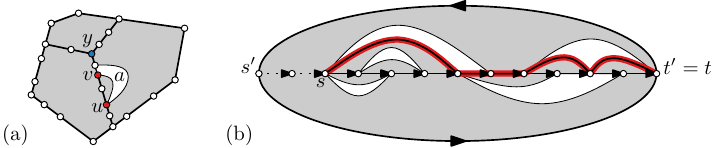}%
  \caption{(a) Splitting $\Gamma^o$ along three straight-line paths.
    The subpolygon containing arch $a$ cannot be extended to a convex
    drawing of its subgraph.  (b)~Left-aligned path $\Lpath_G(P)$ of
    $P=(s,\dots,t)$.}
  \label{fig:leftAligned}
\end{figure}
 
\begin{lemma}[{\hspace{1sp}\cite[Lemma 5, Corollary 6]{DBLP:journals/dam/HongN08}}]\label{lem:leftAligned}
  Let~$G$ be an internally $3$-connected plane graph.  Let
  $P=(s,\dots,t)$ be a subpath of a simple internal directed path $P'$
  between two distinct outer vertices of~$G$.  Then:
  \begin{itemize}[nosep]
  \item $\Lpath_G(P)$ ($\, \Rpath_G(P)$) is a simple internal
    $st$-path
    not arched from the left (right).
  \item If~$P$ is not arched from the right (left) by an internal
    face, then~$\Lpath_G(P)$ ($\, \Rpath_G(P)$) is~archfree.
  \item $\Rpath_G(\Lpath_G(P))$ ($\, \Lpath_G(\Rpath_G(P))$) is
    archfree.
  \end{itemize}
\end{lemma}

\paragraph{Existence of archfree windmills.}

Recall that our plan is to dissect our given (internally)
$3$-connected graph along three archfree paths that form a windmill
pattern; see~\cref{fig:fourRegularAlgo2sketch}a.

\begin{definition}\label{def:windmill}
  Let~$G$ be an internally $3$-connected plane
  graph and let~$o$ denote
  its outer face.  For $i\in[3]$, let $P_i=(o_i,\dots, q_i)$ be a
  simple path in~$G$.  We call $(P_1,P_2,P_3)$ a \emph{windmill}
  of~$G$ if and only if all of the following properties hold (all
  indices are considered modulo $3$):
  \begin{enumerate}[label=(W\arabic*),left=0pt,nosep]
  \item \label{W1} The vertices $o_1,o_2,o_3$ are pairwise distinct
    and belong to $\partial o$.
  \item \label{W2} For $i\in[3]$, no vertex of
    $\V(P_i)\setminus \{o_i\}$ belongs to $\partial o$.
  \item \label{W3} For $i\in[3]$, no interior vertex of $P_i$ belongs
    to $P_{i+1}$.
  \item \label{W4} For $i\in[3]$, the endpoint $q_i$ is an interior
    vertex of $P_{i+1}$.
  \end{enumerate}
  If $(P_1,P_2,P_3)$ is a windmill of $G$, we call it \emph{archfree}
  if $P_1,P_2,P_3$ are archfree.%
\end{definition}

\begin{figure}[tb]
  \centering
  \includegraphics[page=4]{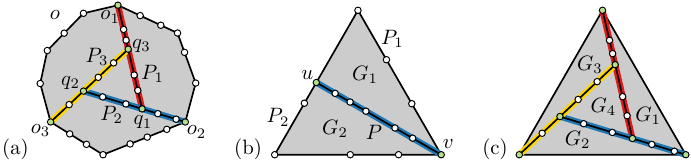}
  \caption{(a) A windmill $(P_1,P_2,P_3)$. (b,c) The $3$-connected
    case in the proof of \cref{thm:fourRegAlgo}.}
  \label{fig:fourRegularAlgo2sketch}
\end{figure}

A necessary condition for the existence of an archfree windmill is the
existence of a \emph{strictly} internal face (a face without outer
vertices).  For the considered graph class we show that the condition
is sufficient.  The following lemma is the main technical contribution
of this section:

\begin{figure}[p]
  \centering
  \includegraphics[height=.97\textheight,page=11]{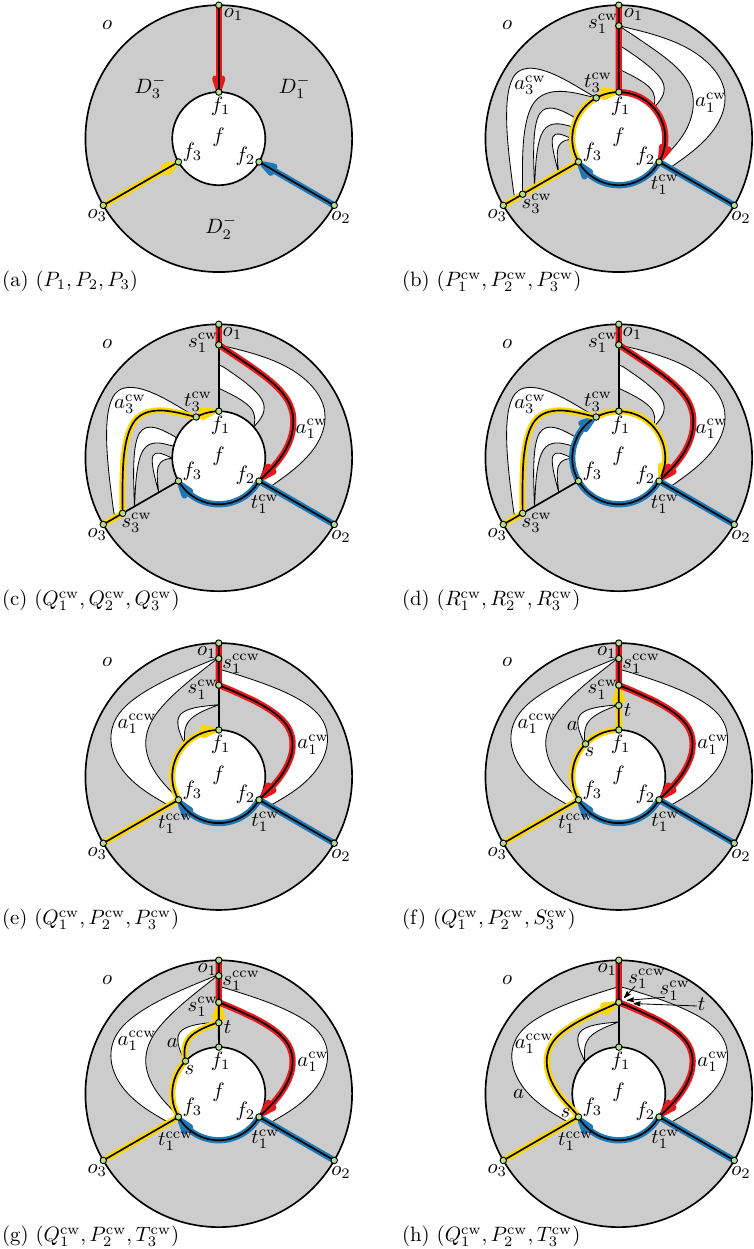}
  \caption{Evolution of the three paths in the first (subfigures
    (a)--(d)) and second (subfigures (e)--(h)) part of the proof of
    \cref{lem:windmill}.
}
  \label{fig:windMillPart1sketch}
\end{figure}

\begin{restatable}[{\hyperref[lem:windmill*]{\appmark}}]{lemma}{windmill}\label{lem:windmill}
  Let $G$ be an internally $3$-connected plane graph of maximum degree
  $4$
  with a strictly internal face~$f$.  Then $G$ contains
  an archfree windmill.
\end{restatable}

\begin{proof}[sketch]
  Let~$o$ be the outer face of $G$.  By means of the internal
  $3$-connectivity of~$G$ and \cref{lem:leftAligned}, it can be shown
  that there are three pairwise disjoint archfree paths
  $P_i=(o_i,\dots, f_i), i\in [3]$ between~$\partial o$
  and~$\partial f$ as depicted in \cref{fig:windMillPart1sketch}a.  We
  now walk along~$\partial f$ in a clockwise fashion and append
  appropriate parts of~$\partial f$ to the paths $P_1,P_2,P_3$ to
  obtain an initial windmill $(P_1^{\cw},P_2^{\cw},P_3^{\cw})$ as
  illustrated in \cref{fig:windMillPart1sketch}b.  Specifically, we
  extend each $P_i$ by the $f_if_{i+1}$ subpath of~$\partial f$ that
  does not contain~$f_{i+2}$ (indices are considered modulo~3).  This
  windmill is not necessarily archfree, but its paths can only be
  arched in a controlled way: suppose that $P_i^{\cw}$ is arched by an
  internal face~$a_i^{\cw}$.  The subpath of $P_i^{\cw}$ that belongs
  to~$\partial f$ is archfree by \cref{lem:minusTwo}.  Combined with
  the fact that~$P_i$ is archfree, it follows that~$a_i^{\cw}$ arches
  $P_i^{\cw}$ between some vertex
  $s_i^{\cw}\in \V(P_i)\setminus \{f_i\}$ and a vertex
  $t_i^{\cw}\in \V(P_i^{\cw})\setminus \V(P_i)$.  Moreover, by
  planarity, $a_i^{\cw}$ has to arch $P_i^{\cw}$ from the left, as
  illustrated in \cref{fig:windMillPart1sketch}b.  We remark that
  there might be multiple ``nested'' faces that arch~$P_i^{\cw}$.
  W.l.o.g., we use $a_i^{\cw}$ to denote the ``outermost'' one, that
  is, the unique arch whose boundary replaces a part of~$P_i^{\cw}$ in
  the left-aligned path $Q_i^{\cw}=\Lpath_G(P_i^{\cw})$, see
  \cref{fig:windMillPart1sketch}c.  The paths of
  $(Q_1^{\cw},Q_2^{\cw},Q_3^{\cw})$ are now archfree by
  \cref{lem:leftAligned}, though, \ref{W4} from \cref{def:windmill}
  is satisfied only for exactly those $i\in [3]$ where the
  arch~$P_{i+1}^{\cw}$ is archfree.  For each $Q_i^{\cw}$ where
  \ref{W4} is violated, we append the $f_{i+1}t_{i+1}^{\cw}$-path
  of~$\partial f$ that does not contain~$f_i$, see
  \cref{fig:windMillPart1sketch}d.  This modification maintains the
  archfreeness by planarity and \cref{lem:minusTwo}.  However, the
  resulting path triple $(R_1^{\cw},R_2^{\cw},R_3^{\cw})$ might still
  not be a windmill: suppose that a path $P_i^{\cw}$ is not archfree
  and its arching face~$a_i^{\cw}$ is \emph{big}, that is,
  $t_i^{\cw}=f_{i+1}$, while additionally the path $P_{i+1}^{\cw}$ is
  archfree (this is the case for $i=1$ in
  \cref{fig:windMillPart1sketch}b).  Then \ref{W3} from
  \cref{def:windmill} is violated for $R_{i+1}^{\cw}$ and \ref{W4}
  is violated for $R_{i+2}^{\cw}$.  Suppose that
  $(R_1^{\cw},R_2^{\cw},R_3^{\cw})$ is indeed not a windmill.  We
  construct path triples $(P_1^{\ccw},P_2^{\ccw},P_3^{\ccw})$,
  $(Q_1^{\ccw},Q_2^{\ccw},Q_3^{\ccw})$, and
  $(R_1^{\ccw},R_2^{\ccw},R_3^{\ccw})$ in a symmetric fashion by
  walking around~$\partial f$ in counterclockwise direction.  If
  $(R_1^{\ccw},R_2^{\ccw},R_3^{\ccw})$ is also not a windmill, it
  follows that both $(R_1^{\cw},R_2^{\cw},R_3^{\cw})$ and
  $(R_1^{\ccw},R_2^{\ccw},R_3^{\ccw})$ contain a path that is arched
  by a big face.  By planarity and the degree bounds, we can now argue
  that there is exactly one $i\in [3]$ such that both
  $(P_i^{\cw})$ and $(P_i^{\ccw})$ are arched by big faces while both
  $(P_{i+1}^{\cw})$ and $(P_{i+2}^{\ccw})$ are archfree, which is
  illustrated in \cref{fig:windMillPart1sketch}e for $i=1$.  Assume
  w.l.o.g.\ that $i=1$ and that~$s_1^{\cw}$ is not closer to~$o_1$
  on~$P_1$ than~$s_1^{\ccw}$.  In view of the previous observations,
  it is now easy to argue that the paths of
  $(Q_1^{\cw},P_2^{\cw},P_3^{\cw})$ are archfree and satisfy all
  windmill properties with the exception of \ref{W4} for $i=3$.  We
  restore \ref{W4} by appending the $f_1s_1^{\cw}$-subpath of~$P_1$
  to $P_3^{\cw}$, see \cref{fig:windMillPart1sketch}f.  By means of
  the degree bounds, it can be argued that \ref{W2} and~\ref{W4}
  are maintained for $i=3$.  The resulting path~$S_3^{\cw}$ might now
  be arched (from the left, by planarity), which can be remedied by
  applying \cref{lem:leftAligned}, see
  Figures~\ref{fig:windMillPart1sketch}g and~h.  By means of the
  degree bounds and planarity arguments, it can be shown that this
  modification maintains all windmill properties.
\end{proof}

A plane graph $G$ is \emph{internally 4-regular} if all of its
internal vertices have degree~$4$ and its outer vertices have degree
at most $4$.  In \cref{lem:windmill}, we established that the
existence of an internal face suffices for the existence of an
archfree windmill.  By means of simple counting arguments, it can be
shown that this condition is satisfied if $G$ has a triangular outer
face.

\begin{restatable}[{\hyperref[lem:strictlyInternalFace*]{\appmark}}]{lemma}{strictlyInternal}
  \label{lem:strictlyInternalFace}
  Let $G$ be an internally 3-connected plane graph that is internally
  4-regular. Let $o$ denote the outer face of $G$ and assume
  $|\partial o|=3$.
  Then $G$ has a strictly internal face.
\end{restatable}

\iftoggle{saveSpace}{}{\paragraph{Algorithm.}

We are now ready to describe our algorithm.
As already mentioned in the beginning of
  \cref{sub:4-regular}, we follow the idea of the recursive
  combinatorial constructions described by Hong and
  Nagamochi~\cite{DBLP:journals/jda/HongN10} and
  Klemz~\cite{DBLP:conf/esa/Klemz21}, though, the way in which we
  decompose our graphs is necessarily quite different.\bk{this
    paragraph is shortened by toggling `saveSpace'} }%

\begin{restatable}[{\hyperref[thm:fourRegAlgo*]{\appmark}}]{theorem}{fourRegAlgo}\label{thm:fourRegAlgo}
  Let~$G$ be an internally $3$-connected internally $4$-regular plane
  graph and let~$\Gamma^o$ be a compatible convex drawing of its outer
  face.  There exists a convex drawing~$\Gamma$ of~$G$ that
  uses~$\Gamma^o$ as the realization of the outer face where each
  internal vertex of~$G$ is contained in the interior of some segment
  of~$\Gamma$.
\end{restatable}

\begin{proof}[sketch]
  Our goal is to (recursively) compute coordinates for the internal
  vertices to obtain the desired drawing of~$G$.
  The base case of the recursion is that~$G$ contains no internal
  edges, in which case there is nothing to show.
  Assume that~$G$ is $3$-connected~-- we deal with the case where~$G$
  is not $3$-connected in the \lncsarxiv{full version}{appendix}.
  If $|\V(\Gamma^o)|\ge 4$, then there exist two distinct outer
  vertices~$u,v$ that do not belong to a common segment of~$\Gamma^o$,
  see \cref{fig:fourRegularAlgo2sketch}b.  By $3$-connectivity and
  \cref{lem:leftAligned}, they are joined by an archfree internal
  path~$P$.  We split $\Gamma^o$ into two simple convex polygons
  along~$P$ and handle the two corresponding subgraphs recursively.
  If $|\V(\Gamma^o)|=3$, then $G$ contains an archfree windmill
  $(P,S,Q)$ by \cref{lem:strictlyInternalFace,lem:windmill}.  Since
  the three outer endpoints of $P,S,Q$ do not belong to a common
  segment of~$\Gamma^o$, we can embed them in a straight-line fashion
  such that~$\Gamma^o$ is dissected into four simple convex polygons,
  see \cref{fig:fourRegularAlgo2sketch}c.  We handle the corresponding
  four subgraphs recursively.
\end{proof}

\paragraph{Universal upper bound.}

Recall \iftoggle{saveSpace}{}{(from the beginning of \cref{sub:4-regular})} that to establish
the claimed upper bound, it suffices to create a drawing where all but
three of the vertices of the graph are drawn in the interior of some
segment.  To achieve this goal, we can now draw the outer face of the
graph as a triangle and then apply \cref{thm:fourRegAlgo}.

\begin{restatable}[{\hyperref[thm:fourUnivUpper*]{\appmark}}]{theorem}{fourUnivUpper}\label{thm:fourUnivUpper}
  Every $3$-connected internally $4$-regular plane graph~$G$ admits a
  convex drawing on at most $n+3$ segments where $n$ is the number of
  vertices.
\end{restatable}

\paragraph{Existential
  lower bound.}

For a graph~$G$, let $G^2$ denote the \emph{square} of~$G$, that is,
\iftoggle{saveSpace}{$\V(G)=\V(G^2)$,}{$G^2$ has the same vertex set as~$G$}
and two vertices in~$G^2$ are
adjacent if and only if their distance in~$G$ is at most~2.
For $n\ge6$, the square of the $n$-cycle, $C_n^2$, is 4-regular
and triconnected.
By removing three edges from a drawing $\Gamma$ of $C_n^2$, we obtain
a drawing of a graph whose segment number is $n$~\cite[proof of
Thm.~7]{desw-dpgfss-CGTA07}. Consequently, $\Gamma$ uses at least
$n-3$ segments\iftoggle{saveSpace}{. We}{, which already shows that
  \cref{thm:fourUnivUpper} is tight up to an additive constant.
  In \lncsarxiv{the full version}{\cref{app:4regular}}, we examine the situation more closely to}
prove a slightly stronger bound.

\begin{restatable}[{\hyperref[prop:square*]{\appmark}}]{proposition}{propSquare}
  \label{prop:square}
  For even $n\ge6$, $C_n^2$ is planar and $\seg(C_n^2) \ge n$.
\end{restatable}

\iftoggle{saveSpace}{}{It's easy to show a slightly worse bound.
  Consider the outerpath $R_n$ where every vertex has degree at
  most~4.  By adding three edges to~$R_n$, we obtain $C_n^2$.
  Dujmovi\'c et al.~\cite{desw-dpgfss-CGTA07} have shown that
  $\seg(R_n)=n$.  Let $\Gamma$ be a drawing of~$R_n$ with $n$
  segments.  Each time we insert one of the three missing edges
  into~$\Gamma$, we can remove at most two ports, hence
  $\seg(C_n^2) \ge n-3$.

  Recall that~$C_6^2$ is the octohedron.  It is known that
  $\seg(C_6^2)=9$ \cite{krw-dgfcf-JGAA19}.  Hence, for this graph, the
  bound in \cref{thm:fourUnivUpper} is best possible.
}%

\section{Maximal Outerpaths}
\label{sec:outerpaths}

In this section, we generalize segments and arcs to
pseudo-$k$-arcs (defined below)
and give a universal lower bound for the number of
pseudo-$k$-arcs in drawings of maximal outerpaths.

We call a
sequence $v_1, v_2, \dots, v_n$ of the vertices of a maximal outerpath~$G$
a \emph{stacking order} of $G$ if for each $i$, the graph $G_i$
induced by the vertices $v_1, v_2, \dots, v_i$ is a maximal outerpath.
An arrangement of
\emph{pseudo-$k$-arcs} is a set of curves in the plane such that any
two of the curves intersect at most $k$ times.  (If two curves share a
tangent, this counts as two intersections.)  We forbid
self-intersections, but for $k \ge 2$ we allow a pseudo-$k$-arc to be
closed.

To show the bound, we present a charging scheme that assigns internal
edges to pseudo-$k$-arcs.  Any drawing of a maximal outerpath has exactly $n - 3$
internal edges.  A pseudo-$k$-arc is \emph{long} if it contains at
least $k+1$ internal edges; otherwise it is \emph{short}.  Let
$\arc_k$ denote the number of pseudo-$k$-arcs, and let $\arc_k^i$
denote the number of pseudo-$k$-arcs with $i$ internal edges.  The
internal edges of a long arc $\alpha$ subdivide the outerpath into
subgraphs $H_0, H_1, \dots, H_{\ell}$ called \emph{bays}; see
\cref{fig:definitionOfHs}.
Given a drawing~$\Gamma$ of a maximal outerpath,
we denote the sub-drawings of $G_3, G_4, \dots, G_n$ within $\Gamma$
by $\Gamma_3, \Gamma_4, \dots, \Gamma_n$, respectively.
A pseudo-$k$-arc~$\alpha$ is
\emph{incident} to a face~$f$ if $\alpha$ contains an edge incident to
a vertex of~$f$.  We say that $\alpha$ is \emph{active} in~$\Gamma_i$ if
$\alpha$ is incident to the last face that has been
added.

\begin{figure}[t]
  \centering \includegraphics{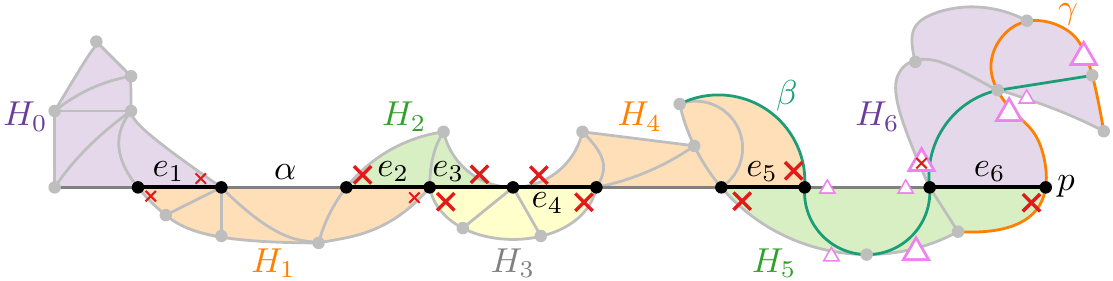}
  \caption{
  	An outerpath represented by a pseudo-2-arc arrangement.
  	The internal edges $e_1, \dots, e_6$ of arc~$\alpha$
  	subdivide the outerpath into bays $H_0, \dots, H_6$.
        We marked the bay crossings of $\alpha$ and $\beta$ by red
        crosses and violet triangles, respectively.  For the bay
        crossings in~$C$ that are relevant for our charging scheme
        we used larger symbols.}
  \label{fig:definitionOfHs}
  \vspace*{-2ex}
\end{figure}

\begin{restatable}[{\hyperref[lem:long*]{\appmark}}]{lemma}{activelong}
  \label{lem:long}
  For any $i \in \{3,\dots,n\}$, a partial outerpath drawing~$\Gamma_i$
  contains at most one active long pseudo~$k$-arc.
\end{restatable}

We do a 2-round assignment to assign each internal edge to a
pseudo-$k$-arc.  We start with the \emph{round-1 assignment}.  Let $I$
denote the set of internal edges of long pseudo-$k$-arcs starting at
the $(k + 1)$-th internal edge (as for the first $k$ internal edges an
arc is still short).  We assign all $n-3$ internal edges except for
the edges in~$I$ to their own pseudo-$k$-arcs:
\begin{equation}
  \label{eq:round-1-assignment}
  (n - 3) - |I| = k \arc_k^{\ge k} + (k-1) \arc_k^{k-1} + \dots + \arc_k^1
  = k \arc_k - \textstyle{\sum_{i=0}^{k}} (k - i) \arc_k^i
\end{equation}

Now we describe the \emph{round-2 assignment}.  There, we charge
the internal edges of~$I$ to specific crossings,
which we can charge in turn to pseudo-$k$-arcs.
A \emph{crossing} is a triple $(\alpha, \beta, p)$ that
consists of two pseudo-$k$-arcs $\alpha$ and $\beta$ and a point $p$
at which $\alpha$ and $\beta$ intersect.
These specific crossings involve long arcs
and we call them \emph{bay crossings}.
Next, we define them such that for each long pseudo-$k$-arc~$\alpha$
with $\ell$ internal edges ($\ell > k$),
there are $2 \ell$ bay crossings ($\alpha$, $\ast$, $\ast$)
where $\ast$ is a wildcard.
For each bay $H \in \{H_1, \dots, H_{\ell - 1} \}$,
we have two bay crossings:
a crossing of $\alpha$ with another pseudo-$k$-arc
at each of the two vertices of~$H$ that have degree~2 within~$H$;
see the red crosses in \cref{fig:definitionOfHs}.
Clearly, they exist for each~$H$ because $H$ is an outerpath.
Since these two vertices are distinct for each pair of consecutive bays,
their bay crossings are distinct as well.
Note that a tangential point may be shared by some $H_j$ and $H_{j + 2}$ (for
$j \in [\ell-3]$); see, e.g., $H_2$ and $H_4$ in
\cref{fig:definitionOfHs}.  However, we still have distinct bay crossings
for $H_j$ and $H_{j + 2}$ since a tangential point counts for two crossings.
For each of $H_0$ and $H_{\ell}$, there is one bay crossing defined next.
In $H_0$ and $H_{\ell}$, consider the two crossings
of~$\alpha$ at the internal edge $e_1$ and $e_{\ell}$, respectively~--
one at each of the vertices of the internal edge.
One of these vertices is the degree-2 vertex of $H_1$ ($H_{\ell - 1}$)
and hence may be identical with
a bay crossing of $H_1$ ($H_{\ell - 1}$).
E.g., in \cref{fig:definitionOfHs},
the bay crossing $(\alpha, \gamma, p)$ of $H_5$
occurs as one of the considered crossings of $H_6$.
The other one of the two considered crossings cannot be a bay crossing
in a neighboring bay and this is our bay crossing of~$H_0$ ($H_\ell$);
see the red crosses at $H_0$ and $H_6$ in \cref{fig:definitionOfHs}.

In the round-2 assignment, we charge the
surplus internal edges of a long arc~$\alpha$ to the other pseudo-$k$-arcs
involved in bay crossings with $\alpha$.
For each internal edge of $I$, we have two distinct
bay crossings of the preceding bay, e.g., in \cref{fig:definitionOfHs}
$H_2$ provides two bay crossings for $e_3$.
Let~$C$ be the set of these
bay crossings.  The bay crossings of
$H_0, \dots, H_{k - 1}$, and $H_{\ell}$ are not included in~$C$
as the internal edges $e_1, \dots, e_k$ are not contained in~$I$
and there is no $e_{\ell + 1}$.
Clearly, $2\abs{I} = \abs{C}$.

Next, we give an upper bound for $|C|$ in terms of $\arc_k$.  The main
argument we exploit is that, by definition, each pseudo-$k$-arc can
participate in at most $k$ crossings with the (current) long arc
and, hence, also in at most $k$ bay crossings with the (current) long arc.
However, we need to be a bit careful when one long pseudo-$k$-arc
becomes inactive and a new pseudo-$k$-arc becomes long, i.e.,
we consider the transition between one long arc to a new long arc.
A (not necessarily long) pseudo-$k$-arc $\gamma$ could potentially
contribute $k$ crossings in $C$ with each long arc.
To compensate for the double
counting at transitions, we introduce the \emph{transition loss}~$t_k$,
which we define as $t_k = \sum_{\gamma \in \mathcal{A} \setminus \set{\alpha_1} }(|\set{c = (\ast, \gamma, \ast) \mid c \in C}|-k)$,
where $\mathcal{A}$ is the set of all pseudo-$k$-arcs
and $\alpha_1$ is the first long arc in $\Gamma$.
In other words, each pseudo-$k$-arc, while it is short, contributes to $t_k$
the number of its bay crossings minus $k$.
For example, in \cref{fig:definitionOfHs},
$\gamma$ contributes $1$ to $t_k$:
$\gamma$ has one bay crossing in~$C$ with the long
arc $\alpha$ (red cross at $e_6$) and two bay crossings in~$C$ with the
long arc $\beta$ (violet triangles on the top right).
The arc $\beta$ contributes $-1$ to $t_k$:
$\beta$ has one bay crossing in~$C$ with the long arc $\alpha$.

Note that, while it is long, an arc does not cross other long arcs.
Also, we do not count the crossings of the
first $k$ bays and the very last bay.  Hence,
\begin{equation}
  \label{eq:crossing-points-at-Hs}
  2\abs{I} = \abs{C} \le
  \hspace{-4.3em}
  \underbrace{k \cdot (\arc_k}_{\parbox{13em}{\scriptsize \centering
      Each pseudo-$k$-arc intersects the current long arc
      at most $k$ times.}} 
  \hspace{-12em}
  \overbrace{-1)}^{\parbox{16em}{\scriptsize \centering The first long
      pseudo-$k$-arc does not provide crossings with another long
      pseudo-$k$-arc.}} 
  \hspace{-4.3em}
  \underbrace{- (2k-1)}_{\parbox{12em}{\scriptsize \centering
      Crossings of $H_0, H_1, \dots, H_{k-1}$ of the first long
      arc are not in $C$.}} 
  \hspace{-7.3em}
  \overbrace{-1}^{\parbox{11.5em}{\scriptsize \centering
      The crossing of $H_{\ell}$ of the last long arc is
      not in $C$.}} 
  \hspace{-4em}
  \underbrace{+ \, t_k}_{\parbox{4em}{\scriptsize \centering transition loss}}
\end{equation}
Plugging \cref{eq:crossing-points-at-Hs} into
\cref{eq:round-1-assignment}, we obtain the following general formula,
which gives a lower bound on the number of pseudo-$k$-arcs for any
outerpath.%
\begin{equation}
  \label{eq:general-formula-pseudo-k-arcs}
  \arc_k \ge \big(2n - 6 + 2 \cdot \textstyle{\sum_{i=0}^{k}} (k - i)
    \arc_k^i - t_k\big)/(3k) + 1 
\end{equation}

Since this formula still contains unresolved variables, we now
resolve~$t_k$.

\begin{restatable}[{\hyperref[clm:transition-loss-k=2*]{\appmark}}]{lemma}{transitionlossarcs}
  \label{clm:transition-loss-k=2}
  There is a loss of at most one crossing per transition from one long
  pseudo-$k$-arc to another long pseudo-$k$-arc.
  Hence, $t_k \le \max \{0, \arc_k^{>k} - 1\} \le \arc_k^{>k} = \arc_k -
  \sum_{i=0}^k \arc_k^i$, where $\arc_k^{>k}$ is the number of long
  pseudo-$k$-arcs.
\end{restatable}
\noindent By \cref{clm:transition-loss-k=2} and
\cref{eq:general-formula-pseudo-k-arcs},
\begin{equation}
  \label{eq:general-formula-pseudo-k-arcs2}
  \arc_k \ge \big(2n+3k-6 + \textstyle{\sum_{i=0}^{k}} (2k-2i+1)
  \arc_k^i\big)/(3k+1) \, .
\end{equation}
Into this general formula,
we plug specific values of $k$ and prove lower bounds on~$\arc_k^i$.
We start with $k = 1$, i.e., outerpath drawings on pseudo
segments.

\begin{restatable}[{\hyperref[clm:outerpath-has-2-zero-segs-and-3-one-segs*]{\appmark}}]{lemma}{outerpathZeroOneSegs}
  \label{clm:outerpath-has-2-zero-segs-and-3-one-segs}
  For $k = 1$ and $n \ge 3$, in any outerpath drawing either
  $\arc_1^0 \ge 3$ or ($\arc_1^0 \ge 2$ and $\arc_1^1 \ge 3$).
\end{restatable}

Using \cref{clm:outerpath-has-2-zero-segs-and-3-one-segs}, we
fill the gaps in \cref{eq:general-formula-pseudo-k-arcs2} for $k = 1$
and obtain~\cref{clm:outerpath-segs}.

\begin{restatable}[{\hyperref[clm:outerpath-segs*]{\appmark}}]{theorem}{outerpathSegs}
	\label{clm:outerpath-segs}
	For any $n$-vertex maximal outerpath $G$,
	$\seg(G) \ge \floor{\frac{n}{2}} + 2$.
\end{restatable}

For $k = 2$, i.e, for (pseudo) circular arcs,
\cref{eq:general-formula-pseudo-k-arcs2} leads to the following
bound.

\begin{restatable}[{\hyperref[clm:outerpath-arcs*]{\appmark}}]{theorem}{outerpathArcs}
	\label{clm:outerpath-arcs}
	For any $n$-vertex maximal outerpath $G$,
	$\arc(G) \ge \ceil{\frac{2n}{7}}$.
\end{restatable}

For $k > 2$, it is not obvious how to generalize circular arcs.
Still, we can make a similar statement for curve arrangements, which
follows directly from \cref{eq:general-formula-pseudo-k-arcs2}.

\begin{proposition}
  \label{clm:outerpath-curves}
  Let $G$ be an $n$-vertex maximal outerpath drawn on a curve
  arrangement in the plane s.t.\ curves intersect pairwise
  $\le k$ times, can be closed, but do not self-intersect.  Then,
  the number~$\arc_k(G)$ of curves required is
  $\ceil{\frac{2n+3k-6}{3k+1}}$.
\end{proposition}

The infinite families of examples in \cref{prop:outerpathExamples}
and \cref{fig:tightOuterpaths}
show that our bounds for segments and arcs are
tight.  This implies, somewhat surprisingly, that, at least for
worst-case instances, using pseudo segments requires as many elements
as using straight line segments.  Whether this also holds for pseudo
circular arcs and circular arcs is an open question.  With circular
arcs, we could not beat a bound of $n/3$, which we could do for pseudo
circular arcs.

\begin{figure}[tb]
	\centering
	\includegraphics{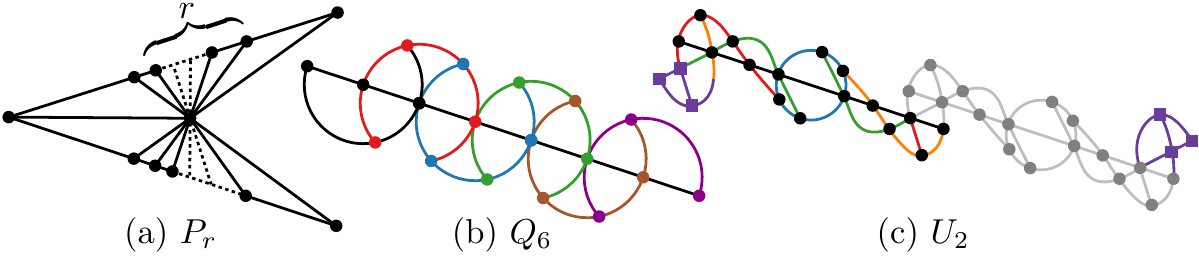}
	
	\caption{Families of maximal outerpaths with (a) $n/2 + 2$
		segments (matching the lower bound in
		\cref{clm:outerpath-segs}),
		(b) $n/3 + 1$ circular arcs, and (c) $(5n+18)/16 < n/3$ pseudo
		2-arcs.  }
	\label{fig:tightOuterpaths}
\end{figure}

\begin{restatable}[{\hyperref[clm:outerpathExamples*]{\appmark}}]{proposition}{outerpathExamples}
  \label{prop:outerpathExamples}
  For every $r\in\mathbb{N}$, maximal outerpaths~$P_r$,
  $Q_r$, $U_r$ exist s.t.\
  \begin{enumerate}[left=0pt,label=(\roman*),nosep]
  \item $P_r$ has $n=2r+6$ vertices and
    $\seg(P_r) \le r+5 = n/2 + 2$,
  \item $Q_r$ has $n=3r$ vertices and
    $\arc(Q_r) \le r+1 = n/3 + 1$,
  \item $U_r$ has $n=16r+6$ vertices and
    $\arc_2(U_r) \le 5r+3 = \frac{5n+18}{16} \approx 0.3125 \, n$.
  \end{enumerate}
\end{restatable}

\section{Further Results and Open Problems}
\label{sec:ratio}

In \lncsarxiv{the full version~\cite{arxiv2022}}{\cref{app:outerplanar}}, we give an alternative proof for
\cref{clm:outerpath-segs},
charging segment ends
to vertices.
We also give universal lower bounds on the segment numbers
of 2-trees and maximal outerpaths.
The key idea is to ``glue'' outerpaths,
while adjusting the charging scheme.
With a different charging scheme from segment ends to faces,
we show an (almost) tight universal lower bound for planar 3-trees.

\begin{restatable}[{\hyperref[clm:outerplanarSec*]{\appmark}}]{theorem}{outerplanarSec}
	\label{clm:outerplanarSec}
	For a 2-tree (or a maximal outerplanar graph) $G$ with $n$ vertices,
	$\seg(G) \ge (n+7)/5$.
\end{restatable}

\begin{restatable}[{\hyperref[clm:3treeseg*]{\appmark}}]{theorem}{threetreeseg}
	\label{clm:3treeseg}
	For a planar 3-tree $G$ with $n \geq 6$ vertices,
	$\seg(G) \geq n + 4$.
\end{restatable}

For cactus graphs, we can compute the segment number in linear time.

\begin{restatable}[{\hyperref[thm:cactus-algo*]{\appmark}}]{theorem}{cactusAlgo}
	\label{thm:cactus-algo}
	Given a cactus graph~$G$, we can compute
	$\seg(G)$ in linear time.  Within this timebound, we can
	draw~$G$ using $\seg(G)$ many segments.  If~$G$ is given with an
	outerplanar embedding, the drawing will respect the given embedding.
\end{restatable}

Now we turn to open problems.
The most prominent one is to close the gaps in
\cref{tab:results}.
Since circular-arc drawings are a generalization of
straight-line drawings, it is natural to ask about the maximum
ratio between the segment number and the arc number of a graph.
We make some initial observations regarding this
question in \lncsarxiv{\cite{arxiv2022}}{\cref{app:discussion}}.
Finally, what is the complexity of deciding whether the arc number of
a given graph is strictly smaller than its segment number?

\bibliographystyle{abbrvurl}
\bibliography{abbrv,visual-complexity,segment-number}

\clearpage
\section*{Appendix}
\appendix
\label{sec:appendix}

In the following, we provide full proofs and omitted content.
First, we introduce some notation that we use throughout the appendix.

Recall that a \emph{cactus} is a connected graph where any two simple cycles
share at most one vertex.
A graph $G$ is a \emph{$k$-tree} if it admits a
\emph{stacking order} $v_1, v_2, \dots, v_n$ of the vertices together
with a sequence of graphs $G_{k+1}, G_{k+2}, \dots, G_n=G$ such that
(i)~$G_{k+1}$ is a clique on $\set{v_1, \dots, v_{k+1}}$; and (ii) for
$k+2 \le i$, the graph $G_i$ is obtained from $G_{i-1}$ by
making~$v_i$ adjacent to all vertices of a $k$-clique in $G_{i-1}$.  A
vertex placement in step~(ii) is called a \emph{stacking operation}.
Similarly, we call the sequence of vertices $v_1, v_2, \dots, v_n$ of
a maximal outerplanar graph $G$ its \emph{stacking order} if for each
$i$ the graph $G_i$ induced by the vertices $v_1, v_2, \dots, v_i$ is
a maximal outerplanar graph. If $G$ is an outerpath, each $G_i$ is an
outerpath.

In a straight-line drawing $\Gamma$ of a graph $G$, each segment
terminates at two vertices.  Let $s$ be a segment in~$\Gamma$, and let
$v$ be an endpoint of~$s$.  Geometrically speaking, we could extend
$s$ at $v$ into a face~$f$.  We say that $s$ has a \emph{port} at~$v$
in~$f$.  We call $v$ \emph{open} if $v$ has at least one port and
\emph{closed} otherwise.  Let $\openseg(\Gamma)$ be the number of
ports in~$\Gamma$, and let $\openseg(G)$ be the minimum number of
ports over all straight-line drawings of~$G$.  Observe that, for any
planar graph $G$, it holds that $\seg(G) = \openseg(G)/2$.
Hence, in a drawing of $G$, counting segments is equivalent to
counting ports.

\section{An Algorithm for Cactus Graphs}
\label{app:cacti}

We first state a lower bound for the segment number of cactus graphs.
Then, we give a recursive algorithm that produces drawings meeting the
bound precisely.

\begin{lemma}
  \label{lem:cactus-structure}
  Let $G$ be a cactus graph, let $\eta$ be the number of odd-degree
  vertices of~$G$, and let $\gamma=3c_0+2c_1+c_2$, where $c_i$ is the
  number of simple cycles with exactly $i$ cut vertices
  in~$G$.  Then $\seg(G) \ge \eta/2+\gamma$.
\end{lemma}

\begin{proof}
  If $G$ is a tree, then $\gamma=0$ and $\seg(G)=\eta/2$, as shown by
  Dujmovi\'c et al.~\cite{desw-dpgfss-CGTA07}.  If $G$ is a cycle,
  then all vertices have degree~2 (that is, $\eta=0$).  Moreover,
  $c_0=1$ and $c_1=c_2=0$.  A cycle can be drawn as a triangle (but
  not with less than three segments), that is, $\seg(G)=3$.

  So assume that $G$ is neither a tree nor a cycle.  Then $G$ contains
  at least one cycle and each cycle has at least one
  cut vertex, that is, $c_0=0$.  Let $\Gamma$ be any straight-line
  drawing of~$G$.  Every odd-degree vertex of~$G$ has a port
  in~$\Gamma$.  Hence, $\Gamma$ has at least $\eta$ ports.

  Additionally, each cycle~$f$ of~$G$ is a simple polygon
  in~$\Gamma$.  In other words, $f$ is incident to at least three
  segments in~$\Gamma$.  If~$f$ contains exactly two cut vertices, the
  drawing of~$f$ must contain a bend at some vertex of~$f$ that is not
  a cut vertex, that is, at a degree-2 vertex.  This increases the
  number of ports by~2.  Similarly, if~$f$ contains exactly one cut
  vertex, the drawing of~$f$ must contain two bends at degree-2
  vertices, which increases the number of ports by~4.  In total,
  $\Gamma$ has at least $\eta+4c_1+2c_2$ ports or $\eta/2+2c_1+c_2$
  segments.  Since $c_0=0$, we have $\seg(G) \ge \eta/2+3c_0+2c_1+c_2$
  as claimed.
\end{proof}

It is not difficult, but somewhat technical to draw a given cactus
such that the lower bound in the above lemma is met exactly.  For an
idea of how we proceed, refer to \cref{fig:cactus-algo}.

\label{thm:cactus-algo*}
\cactusAlgo*

\begin{proof}
  If $G$ is a tree, we can use the linear-time algorithm of Dujmovi\'c
  et al.~\cite{desw-dpgfss-CGTA07}, which yields a drawing with
  $\eta/2$ segments, which is optimal.  If $G$ is a simple cycle, we
  can draw~$G$ as a triangle, which again is optimal.  Otherwise,
  $c_0=0$.  In this case, which we treat below, we draw~$G$ with
  $\eta/2 + 2c_1 + c_2$ segments, which is optimal according to
  \cref{lem:cactus-structure}.

  We draw~$G$ recursively, treating its biconnected components as
  units.  Note that, in a cactus graph, the biconnected components
  (called \emph{blocks}) are exactly its simple cycles and the edges
  that do not lie on any simple cycle.  The \emph{block-cut tree} of a
  connected graph~$H$ has a node for each cut vertex and a node for
  each block.  A block node and a cut-vertex node are connected by an
  edge in the tree if, in~$H$, the block contains the cut vertex.

  We compute the block-cut tree of~$G$, which can be done in linear
  time~\cite{Tar72}, and root it at a block node that corresponds to a
  simple cycle~$f$.  We start by drawing this block as a regular
  $p$-gon~$P$, where $p$ is the maximum of~3 and the number of cut
  vertices of~$f$.  Let $2r$ be the edge length of~$P$, and
  let~$\alpha$ be the interior angle at each corner of~$P$.  Then
  $\alpha = 180^\circ\cdot (p-2)/p$.

  \begin{figure}[hb]
    \centering \includegraphics{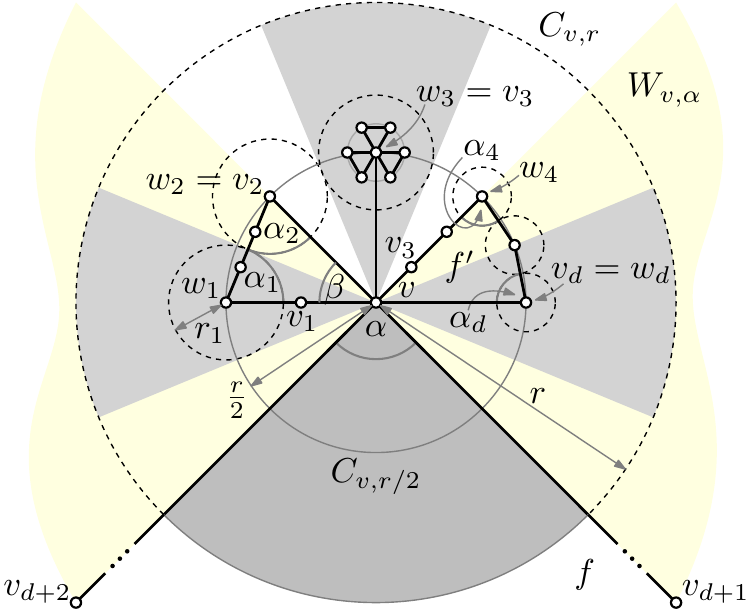}
    \caption{Recursive approach for drawing cactus graphs.  Vertex~$v$
      is a cut vertex of~$f$ (or a degree-2 vertex if $f$ has less
      than three cut vertices).  After~$f$ has been drawn, the
      algorithm recursively draws the subgraph $G(v)$ into~$C_{v,r}$
      such that~$v$ has a port if and only if $\deg(v)$ is odd.}
    \label{fig:cactus-algo}
  \end{figure}

  For each cut vertex~$v$ of~$f$, we recursively draw the
  subgraph~$G(v)$ of~$G$ corresponding to the subtree that hangs
  off~$v$ in the block-cut tree; see \cref{fig:cactus-algo}.  We
  draw~$G(v)$ into the interior of the circle $C_{v,r}$ of radius~$r$
  centered at~$v$.  (Within this circle, we use only the complement
  of~$P$.)  For each pair of cut vertices, the interiors of the
  corresponding circles are disjoint; hence, the drawing of~$G$ has no
  edge crossings if the drawings of the subgraphs are crossing-free.
  Our drawing of~$G$ will have the following property.  Each
  odd-degree vertex has exactly one port and, in every simple cycle
  of~$G$ with $j<3$ cut vertices, there are exactly $3-j$ degree-2
  vertices with two ports.  This implies that the total number of
  segments in our drawing meets the bound in
  \cref{lem:cactus-structure} precisely.

  Let $d = \deg(v)-2$, and let $\N(v) = \set{v_1, \dots, v_{d+2}}$ be
  the neighborhood of~$v$.  Let $v_{d+1}$ and $v_{d+2}$ be the two
  neighbors of~$v$ that lie on~$f$ (in clockwise order before and
  after~$v$ on~$f$) and have already been placed.  Let
  $v_1, \dots, v_{d+2}$ be ordered clockwise around~$v$.  We assume
  that neighbors that belong to the same simple cycle are consecutive
  in this ordering.  (Note that this is the case if $G$ is given with
  a fixed outerplane embedding.)  We now define a set~$W$ of vertices
  in $G(v)$ for which we may call our algorithm recursively.
  Initially, $W$ is empty.  For $i \in \set{1, \dots, d}$, if $v_i$
  and $v$ do not lie on the same simple cycle, then set $w_i = v_i$
  and add~$w_i$ to~$W$.  Now let $f'$ be the simple cycle that
  contains~$v$, $v_i$ and another neighbor of~$v$, say,~$v_{i+1}$.  If
  $f'$ does not contain a cut vertex other than~$v$, set $w_i = v_i$
  and add~$w_i$ to~$W$.  Otherwise, let $w_i$ be the cut vertex of~$G$
  closest to~$v_i$ in $G(v)-v$.  If $v_{i+1}$ has the same closest cut
  vertex~$w_j$ then, if $w_i \ne v_i$, set $w_i = v_i$, otherwise set
  $w_{i+1} = v_{i+1}$.  Add~$w_i$ and~$w_{i+1}$ and all other cut
  vertices of~$f'$ (if any) to~$W$ (except~$v$).

  We now place the vertices in~$W$ on the circle~$C_{v,r/2}$.  If~$d$
  is odd, then we place~$w_{(d+1)/2}$ on the line that bisects the
  angle~$\angle v_{d+2}vv_{d+1}$; namely such that~$w_{(d+1)/2}$ lies
  opposite of this angle (as $w_3$ in \cref{fig:cactus-algo}).  For
  the remainder of this proof, we assume for simplicity that $d$ is
  even.  Then $d\ge2$ and we place~$w_{d/2}$ on the line~$vv_{d+1}$
  and~$w_{d/2+1}$ on the line~$vv_{d+2}$.  We place the remaining
  neighbors in pairs on opposite sides of lines through~$v$ such that
  these lines equally partition the angle space in the double
  wedge~$W_{v,\alpha}$ (light yellow in \cref{fig:cactus-algo}) that
  is bounded by the lines~$vv_{d+1}$ and~$vv_{d+2}$ and does not
  contain the angle~$\angle v_{d+2}vv_{d+1}$.  The angular distance
  between two consecutive edges incident to~$v$ is then
  $\beta = (360^\circ - 2\alpha)/d$.

  We draw each simple cycle~$f'$ that contains~$v$ and two neighbors~$v_i$
  and~$v_{i+1}$ of~$v$ as a simple polygon that connects~$v$ to~$w_i$
  to potential further cut vertices of~$f'$ (in their order
  along~$f'$) to~$w_{i+1}$ to~$v$.

  Now we define, for each newly placed vertex~$w \in W$ with
  $\deg(w)>2$, values~$\alpha'$ and~$r'$ so that we can draw the
  graph~$G(w)$ recursively.  To this end, if $v$ and $w$ lie on the
  same simple cycle~$f'$, let~$\alpha'$ be the interior angle of~$w$
  in~$f'$, and let~$r'$ be the distance of~$w$ to the closest vertex
  in~$\V(f)\cap W$ divided by~2.  Otherwise, let $\alpha'$ be~0 and
  set $r'$ such that $C_{w,r'}$ fits into a wedge centered at~$w$ that
  has an angle of~$\beta$ at its apex~$v$; see, for example, $w_3$ in
  \cref{fig:cactus-algo}.

  Our invariant is that, in each recursive call for $G(w')$, we have
  $0\le\alpha'<180^\circ$ and $r'>0$.  This ensures that our drawing
  has no crossings.  To finish the proof, note that the segments that
  we draw end only in odd-degree vertices (one port each) or in
  degree-2 vertices (two ports each) of simple cycles that have less than
  three cut vertices.

  Concerning the running time, it is easy to see that each recursive
  call of the algorithm runs in time linear in the size of the
  subgraph of~$G$ that the current call draws without further
  recursion.  Hence, the overall running time is linear in the size
  of~$G$ (including the computation of the block-cut tree).
\end{proof}

Note that the algorithm in the proof of \cref{thm:cactus-algo} can
draw a cactus with a fixed outerplane embedding such that its
embedding is maintained.  Unfortunately, the drawing area can be at
least exponential, even if the embedding is not fixed.

\section{Proofs Omitted in Section \ref{sub:4-regular} (4-Regular
  Planar Graphs)}
\label{app:4regular}

\label{obs:internally3conCycle*}
\intThreeConCycle*

\begin{proof}%
  Clearly, Property~\ref{I2} of \cref{def:int3conDefs} carries over
  from $G$ to $C^-$.
\end{proof}

\label{lem:windmill*}
\windmill*

\begin{proof}%
  Let~$o$ be the outer face of $G$.  We begin by constructing three
  disjoint archfree paths between~$\partial o$ and~$\partial f$, as
  illustrated in \cref{fig:windMillPart1}a.  We plan to use these
  paths, as well as parts of~$\partial f$ to construct a windmill (see
  \cref{fig:windMillPart1}b).  We then apply \cref{lem:leftAligned} to
  make its paths archfree (illustrated in \cref{fig:windMillPart1}c
  and \cref{fig:windMillPart2}a)).  This may destroy the windmill
  properties, but it does so in a controlled way, which allows us to
  successively modify our paths to restore the windmill properties
  while maintaining the archfreeness.

  \begin{myclaim}\label{claim:intialPaths}
    $G$ contains three simple paths
    $P_i=(o_i,\dots ,f_i), i\in [3]$ such that
    \begin{enumerate}[label=(P\arabic*),left=0pt,nosep]
    \item \label{P1} $P_1,P_2,P_3$ are pairwise vertex-disjoint,
    \item \label{P2} for $i\in[3]$, the endpoint~$o_i$ belongs
      to~$\partial o$,
    \item \label{P3} for $i\in[3]$, the endpoint~$f_i$ belongs
      to~$\partial f$,
    \item \label{P4} for $i\in[3]$, the interior vertices of~$P_i$
      belong to neither~$\partial f$ nor~$\partial o$, and
    \item \label{P5} $P_1,P_2,P_3$ are archfree in $G$.
    \end{enumerate}
  \end{myclaim}

\begin{proof}[of Claim~\ref{claim:intialPaths}]
  For illustrations refer to \cref{fig:windMillPart1}a.  To show
  that~$P_1,P_2,P_3$ exist, we add a new vertex~$v_f$ into~$f$ and add
  edges between~$v_f$ every vertex of~$\partial f$.  It is easy to see
  that this modification retains Property~\ref{I1} of
  \cref{def:int3conDefs} and, hence, the resulting graph $G'$ is
  internally $3$-connected.  By Property~\ref{I2} of
  \cref{def:int3conDefs} applied to~$v_f$ in $G'$, it follows that $G$
  contains three simple paths
  $P_i' = (o_i,\dots ,f_i), i \in \set{1,2,3}$, that satisfy
  Properties~\ref{P1}--\ref{P4}.  For $i = 1,2,3$, we
  consider~$P_i'$ to be directed from~$o_i$ to~$f_i$ and define
  $P_i=\Rpath_G(\Lpath_G(P_i'))$.  By \cref{lem:leftAligned}, the
  paths $P_1,P_2,P_3$ satisfy Property~\ref{P5}.  To see that the
  remaining properties are also satisfied, we argue as follows:
  let $C$ be the simple cycle formed by the paths $P_2,P_3$, the
  $o_2o_3$-path on $\partial o$ that passes through~$o_1$, and the
  $f_2f_3$-path on $\partial f$ that passes through $f_1$.  The closed
  interior~$C^-$ of~$C$ is internally $3$-connected by
  Obs.~\ref{obs:internally3conCycle}.  Note that
  $\Rpath_{C^-}(\Lpath_{C^-}(P_1'))=P_1$ since the only internal face
  of~$G$ that is not an internal face of~$C^-$ but shares a vertex
  with~$P_1'$ is~$f$, which has only a single vertex in common
  with~$P_1'$ by Properties~\ref{P2}--\ref{P4}.  By
  \cref{lem:leftAligned} applied to~$C^-$ and~$P_1'$, $P_1$ is an {\em
    internal} $o_1f_1$-path of~$C^-$.  Hence, the paths
  $P_1,P_2',P_3'$ satisfy Properties~\ref{P1}--\ref{P4} in~$G$.
  With analogous arguments, we see that $P_1,P_2,P_3'$ and, finally,
  $P_1,P_2,P_3$ also satisfy Properties~\ref{P1}--\ref{P4}.
\end{proof}

Without loss of generality, assume that $o_1,o_2,o_3$ appear
on~$\partial f$ in clockwise order, as depicted in
\cref{fig:windMillPart1}a.  For $i\in[3]$, we append the
$f_if_{i+1}$-path on $\partial f$ that does not pass through~$f_{i-1}$
to~$P_i$ and call the resulting simple path~$P_i^{\cw}$ (all indices
are considered modulo $3$), see \cref{fig:windMillPart1}b.
Symmetrically, for $i\in[3]$, we append the $f_if_{i-1}$-path on
$\partial f$ that does not pass through~$f_{i+1}$ to~$P_i$ and call
the resulting simple path~$P_i^{\ccw}$.  By construction, both
$(P_1^{\cw},P_2^{\cw},P_3^{\cw})$ and
$(P_1^{\ccw},P_2^{\ccw},P_3^{\ccw})$ are windmills in~$G$. If one of
them is archfree, we are done. So assume otherwise.

We will now deform parts of $(P_1^{\cw},P_2^{\cw},P_3^{\cw})$ (or
$(P_1^{\ccw},P_2^{\ccw},P_3^{\ccw})$) to obtain the desired archfree
windmill.  To this end, we introduce some notation, for illustrations
refer to \cref{fig:windMillPart1}b:
suppose that a path $P_i^{\cw}, i\in [3]$ is arched by a
face~$a_i^{\cw}$.  The subpath~$P_i$ of $P_i^{\cw}$ is archfree by
Claim~\ref{claim:intialPaths}.  Moreover, the $f_if_{i+1}$-subpath of~$P_i$
is also archfree by \cref{lem:minusTwo}.  Consequently, the
face~$a_i^{\cw}$ arches $P_i^{\cw}$ between some vertex
$s_i^{\cw}\in \V(P_i)\setminus \{f_i\}$ and a vertex
$t_i^{\cw}\in \V(P_i^{\cw})\setminus \V(P_i)$.  We consider
$P_i^{\cw}$ to be directed such that~$o_i$ is its source.  By
planarity, the face~$a_i^{\cw}$ has to arch $P_i^{\cw}$ from the left.
We remark that there might be multiple faces that arch $P_i^{\cw}$
(from the left), in which case these faces have to be ``nested'' (as
depicted in \cref{fig:windMillPart1}b).  Without loss of generality,
we may assume that~$a_i^{\cw}$ is the ``outermost'' of these arches.
More precisely, we assume that $a_i^{\cw}$ is the unique arch such
that a $s_i^{\cw}t_i^{\cw}$-path on $\partial a_i^{\cw}$ replaces the
$s_i^{\cw}t_i^{\cw}$-path on $P_i^{\cw}$ in $\Lpath_G(P_i^{\cw})$.  We
say that~$a_i^{\cw}$ is \emph{big} if $t_i^{\cw}=f_{i+1}$ (in
\cref{fig:windMillPart1}b the arch~$a_1^{\cw}$ is big,
while~$a_3^{\cw}$ is not).  Symmetrically, we define the
expressions~$a_i^{\ccw}, s_i^{\ccw}, t_i^{\ccw}$ for each path
$P_i^{\ccw}, i\in [3]$ that is arched (from the right), and we
say that~$a_i^{\ccw}$ is \emph{big} if~$t_i^{\ccw}=f_{i-1}$.
For $i\in[3]$, let $D_i$ denote the simple cycle that is formed by
$P_i,P_{i+1}$, the $f_if_{i+1}$-path on $\partial f$ that does not
pass through $f_{i+2}$ and the $o_io_{i+1}$-path on $\partial o$
that does not pass through $o_{i+2}$ (the closed interior $D_i^-$ of
$D_i$ is indicated in \cref{fig:windMillPart1}a).

For $i\in[3]$, we define $Q_i^{\cw}=\Lpath_G(P_i^{\cw})$, for
illustrations refer to \cref{fig:windMillPart1}c.

\begin{figure}[tbh]
  \centering
  \includegraphics[page=6]{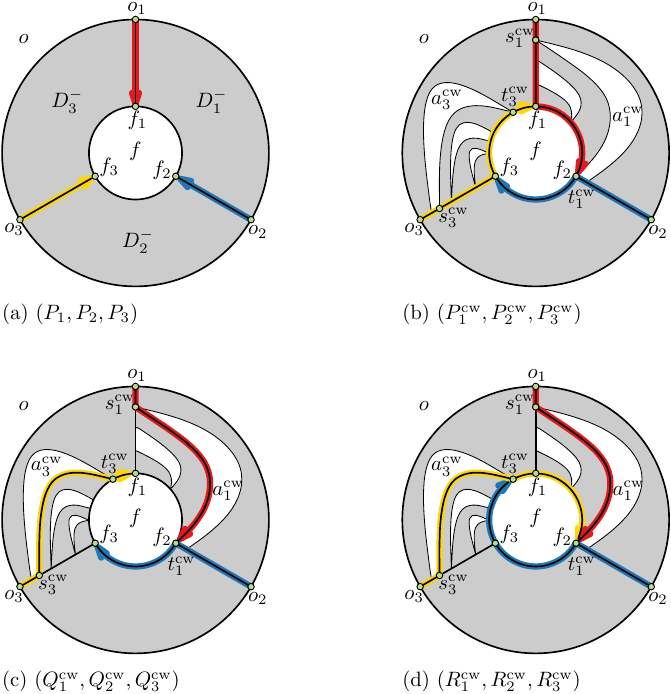}
  \caption{Evolution of the three paths in the first part of the proof
    of \cref{lem:windmill}.}
  \label{fig:windMillPart1}
\end{figure}

\begin{restatable}{myclaim}{pathsQ}\label{claim:pathsQ}
  The paths $Q_1^{\cw},Q_2^{\cw},Q_3^{\cw}$ are
  archfree. Properties~\ref{W1}--\ref{W3} from
  \cref{def:windmill} hold for $(Q_1^{\cw},Q_2^{\cw},Q_3^{\cw})$.
  Property~\ref{W4} from \cref{def:windmill} is satisfied for
  exactly those $i\in [3]$ where the arch~$a_{i+1}^{\cw}$ is
  undefined.
\end{restatable}

\begin{proof}[of Claim~\ref{claim:pathsQ}]
  \cref{lem:leftAligned} implies that the paths are archfree (recall
  that $P_i^{\cw}$ is not arched from the right) and that the
  Properties~\ref{W1} and~\ref{W2} from \cref{def:windmill}
  carry over from $(P_1^{\cw},P_2^{\cw},P_3^{\cw})$ to
  $(Q_1^{\cw},Q_2^{\cw},Q_3^{\cw})$.

  To see that Property~\ref{W3} also carries over, we can argue as
  follows: let $C_i$ be the simple cycle formed by the paths
  $P_{i+1},P_{i+2}$, the $o_{i+1}o_{i+2}$-path on $\partial o$ that
  passes through~$o_i$, and the $f_{i+1}f_{i+2}$-path on $\partial f$
  that does not pass through $f_1$.  The closed interior~$C_i^-$
  of~$C_i$ is internally $3$-connected by
  Obs.~\ref{obs:internally3conCycle}.
  Note that $\Lpath_{C_i^-}(P_i^{\cw})=Q_i^{\cw}$ since the internal
  faces of~$G$ that do not belong to $C_i^-$ intersect $P_i^{\cw}$
  only in~$f_{i+1}$.
  Hence, by \cref{lem:leftAligned} applied to~$C_i^-$ and~$P_i^{\cw}$
  and by construction, the part of $Q_i^{\cw}$ that does not belong to
  $P_i^{\cw}$ is located in the interior of the cycle $D_i$.  By
  construction, the interior of $D_i$ is disjoint from
  $P_1^{\cw},P_2^{\cw}$, and~$P_3^{\cw}$.  Moreover, $D_1,D_2,D_3$ are
  pairwise interior disjoint.  Consequently, Property~\ref{W3} from
  \cref{def:windmill} carries over from
  $(P_1^{\cw},P_2^{\cw},P_3^{\cw})$ to
  $(Q_1^{\cw},Q_2^{\cw},Q_3^{\cw})$, as claimed.

  Finally, Property~\ref{W4} is violated for those $i\in [3]$
  where the arch~$a_{i+1}^{\cw}$ exists since in this case the
  endpoint $f_{i+1}$ of $Q_i^{\cw}$ is not on $Q_{i+1}^{\cw}$.  In
  contrast, Property~\ref{W4} is satisfied for those
  $i\in [3]$ where the arch~$a_{i+1}^{\cw}$ undefined since in
  this case the endpoint $f_{i+1}$ of $Q_i^{\cw}$ is on
  $Q_{i+1}^{\cw}(=P_{i+1}^{\cw})$.
\end{proof}

If $(Q_1^{\cw},Q_2^{\cw},Q_3^{\cw})$ or
$(Q_1^{\ccw},Q_2^{\ccw},Q_3^{\ccw})$ (which is defined symmetrically
and for which a symmetric version of Claim~\ref{claim:pathsQ} holds) is an
archfree windmill, we are done.  So assume otherwise.  By
Claim~\ref{claim:pathsQ}, the only violated property is~\ref{W4}.  To
remedy the situation, we will now construct three paths
$R_1^{\cw},R_2^{\cw},R_3^{\cw}$ by appending appropriate parts
of~$\partial f$ to the corresponding paths in
$(Q_1^{\cw},Q_2^{\cw},Q_3^{\cw})$.

For each $i\in [3]$ where the arch~$a_{i+1}^{\cw}$ is undefined,
we set $R_i^{\cw}=Q_i^{\cw}$.  For each $i\in [3]$ where the
arch~$a_{i+1}^{\cw}$ exists, we append the $f_{i+1}t_{i+1}^{\cw}$-path
on $\partial f$ that does not contain $f_i$ to $Q_i^{\cw}$ and denote
the resulting path by $R_i^{\cw}$, for illustrations refer to
\cref{fig:windMillPart1}d.

\begin{restatable}{myclaim}{pathsR}\label{claim:pathsR}
  The paths $R_1^{\cw},R_2^{\cw},R_3^{\cw}$ are
  archfree. Properties~\ref{W1} and~\ref{W2} from
  \cref{def:windmill} hold for $(R_1^{\cw},R_2^{\cw},R_3^{\cw})$.
  Property~\ref{W3} from \cref{def:windmill} is violated for
  exactly those $i\in [3]$ where $P_{i+2}^{\cw}$ is arched by a
  big face and~$a_i^{\cw}$ is undefined.  Property~\ref{W4} from
  \cref{def:windmill} violated exactly those $i\in [3]$ where
  $P_{i+1}^{\cw}$ is arched by a big face and~$a_{i+2}^{\cw}$ is
  undefined.
\end{restatable}

\begin{proof}[of Claim~\ref{claim:pathsR}]
  Let $i\in [3]$.  To see that $R_i^{\cw}$ is archfree, recall
  that $Q_i^{\cw}$ is archfree by Claim~\ref{claim:pathsQ}.  Hence, if
  $R_i^{\cw}=Q_i^{\cw}$, there is nothing to show, so assume
  otherwise.  The $f_{i+1}t_{i+1}^{\cw}$ subpath of $R_i^{\cw}$ is
  archfree by \cref{lem:minusTwo}.  Therefore, if $R_i^{\cw}$ is
  arched by some internal face~$a\neq f$, then~$a$ has to arch
  $R_i^{\cw}$ between some vertex in
  $\V(R_i^{\cw})\setminus \V(Q_i^{\cw})$ and some vertex in
  $\V(Q_i^{\cw})\setminus \{f_{i+1}\}$, which is impossible by
  planarity (specifically, the boundary of $a$ would have to cross the
  cyle $D_{i+1}$).  Moreover, by construction, $R_i^{\cw}$ is not
  arched by~$f$.  Hence, $R_i^{\cw}$ is archfree, as claimed.

  Clearly, Properties~\ref{W1} and~\ref{W2} of
  \cref{def:windmill} carry over from
  $(Q_1^{\cw},Q_2^{\cw},Q_3^{\cw})$.

  Regarding Property~\ref{W4}, let~$q_j$ denote the (internal)
  endpoint of $R_j^{\cw}$ that is not~$o_j$ for~$j\in[3]$.  By
  construction, $q_i$ belongs to $R_{i+1}^{\cw}$.  Hence,
  Property~\ref{W4} is violated if and only if $q_i$ is not an {\em
    interior} vertex of $R_{i+1}^{\cw}$, which, by construction, is
  the case if and only if $P_{i+1}^{\cw}$ is arched by a big face
  and~$a_{i+2}^{\cw}$ is undefined (in which case
  $q_i=t_{i+1}^{\cw}=f_{i+2}=q_{i+1}$), for an illustration refer to
  \cref{fig:windMillPart1}d with $i=3$.

  Regarding Property~\ref{W3}, we argue in two steps:
  let $I_i^{\mathrm{old}}$ be the set of interior vertices of
  $R_i^{\cw}$ that are also interior vertices of $Q_i^{\cw}$, i.e,
  $I_i^{\mathrm{old}}=\V(Q_i^{\cw})\setminus \{o_i,f_{i+1}\}$.
  Further, let $I_i^{\mathrm{new}}$ be the set of interior vertices of
  $R_i^{\cw}$ that are not interior vertices of $Q_i^{\cw}$.
  Property~\ref{W3} holds if and only if none of these sets
  intersects $\V(R_{i+1}^{\cw})$.

  We first consider the set $I_i^{\mathrm{new}}$.  If
  $R_i^{\cw}=Q_i^{\cw}$, then $I_i^{\mathrm{new}}=\emptyset$ and,
  hence, $I_i^{\mathrm{new}}\cap \V(R_{i+1}^{\cw})=\emptyset$.
  Otherwise,
  $I_i^{\mathrm{new}}=\{f_{i+1}\}\cup (\V(R_i^{\cw})\setminus
  (\V(Q_i^{\cw}\cup t_{i+1})))$.  By construction, this set is
  disjoint from both $\V(Q_{i+1}^{\cw})$ and
  $\V(R_{i+1}^{\cw})\setminus \V(Q_{i+1}^{\cw})$.  Hence,
  $I_i^{\mathrm{new}}\cap \V(R_{i+1}^{\cw})=\emptyset$.

  It remains to consider the set $I_i^{\mathrm{old}}$.  By
  Property~\ref{W3} for $Q_i^{\cw}$, we have
  $I_i^{\mathrm{old}}\cap \V(Q_{i+1}^{\cw})=\emptyset$.  So if
  $I_i^{\mathrm{old}}\cap \V(R_{i+1}^{\cw})\neq\emptyset$, then
  $I_i^{\mathrm{old}}\cap (\V(R_{i+1}^{\cw})\setminus
  \V(Q_{i+1}^{\cw}))\neq\emptyset$.  All vertices in
  $I_i^{\mathrm{old}}$ belong to the closed interior~$D_i^-$ of the
  cycle $D_i$.  The $f_{i+2}t_{i+2}^{\cw}$-path on $\partial f$
  intersects $D_i^-$ if and only if $P_{i+2}$ is arched by a big face
  (i.e., $t_{i+2}^{\cw}=f_i$), namely in~$f_i$.  Hence,
  Property~\ref{W3} is violated for $i$ if and only if
  $P_{i+2}^{\cw}$ is arched by a big face and ~$a_i^{\cw}$ is
  undefined (in which case $Q_i^{\cw}=P_i^{\cw}$ and, hence,
  $f_i\in \V(R_i^{\cw})$).
\end{proof}

If $(R_1^{\cw},R_2^{\cw},R_3^{\cw})$ or
$(R_1^{\ccw},R_2^{\ccw},R_3^{\ccw})$ (which is defined symmetrically
and for which a symmetric version of Claim~\ref{claim:pathsR} holds) is an
archfree windmill, we are done.  So assume otherwise.  By
Claim~\ref{claim:pathsR}, it follows that both
$\{P_1^{\cw},P_2^{\cw},P_3^{\cw}\}$ and
$\{P_1^{\ccw},P_2^{\ccw},P_3^{\ccw}\}$ contain a path that is arched
by a big face.  Specifically, we may assume without loss of generality
that the arch~$a_1^{\cw}$ exists and is big (i.e., $t_1^{\cw}=f_{2}$)
and~$a_2^{\cw}$ is undefined; for an illustration refer to
\cref{fig:windMillPart2}a.  By planarity, the path $P_2^{\ccw}$ cannot
be arched (the boundary of~$a_2^{\ccw}$ would have to cross the
boundary of~$a_1^{\cw}$).  Moreover, the path $P_3^{\ccw}$ cannot be
arched by a big face since this would imply~$\deg(f_2)\ge 5$,
contradicting the degree bounds.  Hence, the arch~$a_1^{\ccw}$ of
$P_1^{\ccw}$ exists and is big (by assumption and
Claim~\ref{claim:pathsR}).
Note that, by planarity, the path $P_3^{\cw}$ cannot be arched (the
boundary of~$a_3^{\cw}$ would have to cross the boundary
of~$a_1^{\ccw}$).
Without loss of generality, we may assume that $s_1^{\cw}$ is not
closer to~$o_1$ on~$P_1$ than~$s_1^{\ccw}$ (otherwise, we can argue
symmetrically), for illustrations refer to
Figures~\ref{fig:windMillPart2}a and~d.

\begin{restatable}{myclaim}{pathsQPP}\label{claim:pathsQPP}
  The paths $Q_1^{\cw},P_2^{\cw},P_3^{\cw}$ are
  archfree. Properties~\ref{W1}--\ref{W3} from
  \cref{def:windmill} hold for $(Q_1^{\cw},P_2^{\cw},P_3^{\cw})$.
  Property~\ref{W4} from \cref{def:windmill} is satisfied for
  $i=1,2$, but violated for $i=3$.
\end{restatable}

\begin{proof}[of Claim~\ref{claim:pathsQPP}]
  Since neither $P_2^{\cw}$ nor $P_3^{\cw}$ are arched,
  Claim~\ref{claim:pathsQ} implies the claimed properties hold for
  $(Q_1^{\cw},Q_2^{\cw},Q_3^{\cw})$.  Moreover, $P_2^{\cw}=Q_2^{\cw}$
  and $P_3^{\cw}=Q_3^{\cw}$, which proves the claim.
\end{proof}

\begin{figure}[tbh]
  \centering \includegraphics[page=7]{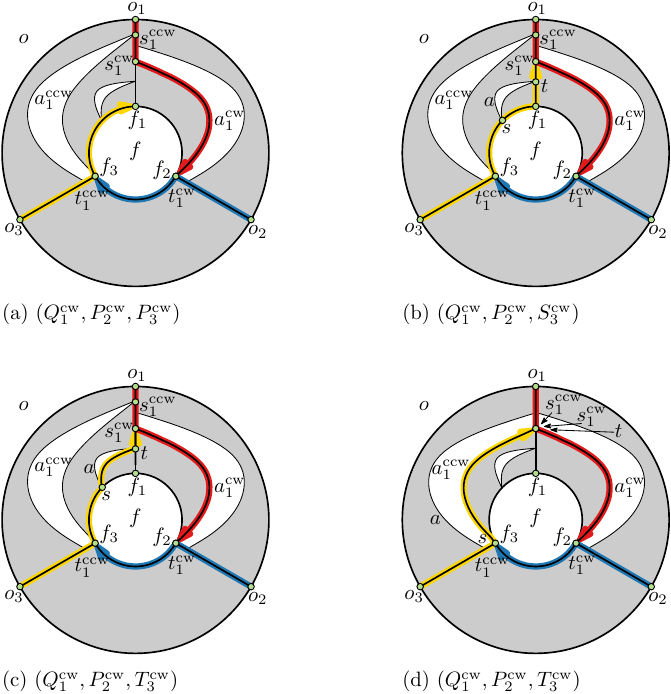}
  \caption{Evolution of the three paths in the second part of the
    proof of \cref{lem:windmill}.}
  \label{fig:windMillPart2}
\end{figure}

We now append the $f_1s_1^{\cw}$-subpath of $P_1$ to $P_{3}^{\cw}$ and
denote the resulting path by $S_{3}^{\cw}$, for an illustration see
\cref{fig:windMillPart2}b.

\begin{restatable}{myclaim}{pathsQPS}\label{claim:pathsQPS}
  The paths $Q_1^{\cw}$ and $P_2^{\cw}$ are archfree.
  Properties~\ref{W1}--\ref{W4} from \cref{def:windmill} hold
  for $(Q_1^{\cw},P_2^{\cw},S_3^{\cw})$.
\end{restatable}

\begin{proof}[of Claim~\ref{claim:pathsQPS}]
  To see that Properties~\ref{W2} and~\ref{W4} hold for $i=3$,
  we argue as follows: towards a contradiction, assume that
  $s_1^{\cw}\in \partial o$, i.e., $s_1^{\cw}=o_1$.  By our
  assumption about the positions of $s_1^{\cw}$ and $s_1^{\ccw}$
  on~$P_1$, it follows that $s_1^{\cw}=s_1^{\ccw}=o_1$.  However, this
  implies that $\deg o_1\ge 5$, contradicting the degree bounds.

  Clearly, by construction, the remaining properties of
  $(Q_{1}^{\cw},P_{2}^{\cw},P_{3}^{\cw})$ guaranteed by
  Claim~\ref{claim:pathsQPP} carry over to
  $(Q_{1}^{\cw},P_{2}^{\cw},S_{3}^{\cw})$.
\end{proof}

If $(Q_1^{\cw},P_2^{\cw},S_3^{\cw})$ is an archfree windmill, we are
done, so assume otherwise.  By Claim~\ref{claim:pathsQPS}, it follows that
$S_3^{\cw}$ is arched by an internal face.  The paths $P_{3}^{\cw}$
and $P_1$ are archfree, so an internal face $a$ that is arching
$S_{3}^{\cw}$ has to do so between some vertex in
$t\in \V(S_3^{\cw})\setminus \V(P_3^{\cw})$ and some vertex in
$s\in \V(P_3^{\cw})\setminus \{f_1\}$.  By planarity, it is not
possible that~$a$ arches $S_{3}^{\cw}$ from the right, so~$a$ arches
$S_{3}^{\cw}$ from the left.  As in the definition of the
arches~$a_i^{\cw}$, there may be multiple nested arches that arch
$S_{3}^{\cw}$ from the left, and we assume~$a$ to be the``outermost''
one.  More precisely, we assume that $a$ is this unique arch such that
a $st$-path on $\partial a$ replaces the $st$-path on $S_{3}^{\cw}$ in
$\Lpath_G(S_{3}^{\cw})$.

As the final modification, we replace $S_{3}^{\cw}$ by
$T_{3}^{\cw}=\Lpath_G(S_{3}^{\cw})$, for an illustration see
Figures~\ref{fig:windMillPart2}c and~d.

\begin{restatable}{myclaim}{pathsQPT}\label{claim:pathsQPT}
  $(Q_1^{\cw},P_2^{\cw},T_3^{\cw})$ is an archfree windmill.
\end{restatable}

\begin{proof}[of Claim~\ref{claim:pathsQPT}]
  By Claim~\ref{claim:pathsQPS}, the paths $Q_1^{\cw}$ and $P_2^{\cw}$ are
  archfree and Properties~\ref{W1}--\ref{W4} from
  \cref{def:windmill} hold for $(Q_1^{\cw},P_2^{\cw},S_3^{\cw})$.  By
  \cref{lem:leftAligned}, $T_{3}^{\cw}$ is archfree.  To conclude the
  proof, it remains to show that Properties~\ref{W1}--\ref{W4}
  hold for the set $(Q_1^{\cw},P_2^{\cw},T_3^{\cw})$.

  Properties~\ref{W1} and~\ref{W2} are preserved from
  $(Q_1^{\cw},P_2^{\cw},S_3^{\cw})$ (by \cref{lem:leftAligned} for
  $i=3$).  To show that the remaining two properties hold, we first
  show that $s\notin \V(P_3)\setminus \{f_3\}$.  If
  $s_1^{\cw}\neq s_1^{\ccw}$ (see \cref{fig:windMillPart2}c), this is
  clear by planarity (the boundary of $a$ would have to cross the
  boundary of $a_1^{\ccw}$).  Towards a contradiction, assume that
  $s_1^{\cw}=s_1^{\ccw}$ (see \cref{fig:windMillPart2}c) and
  $s\in \V(P_3)\setminus \{f_3\}$.  By planarity, it follows that
  $t=s_1^{\cw}=s_1^{\ccw}$ (otherwise, the boundary of $a$ would have
  to cross the boundary of $a_1^{\ccw}$).  However, this implies that
  $\deg(t)\ge 5$, contradicting the degree bounds.  So indeed,
  $s\notin \V(P_3)\setminus \{f_3\}$ as claimed.

  Clearly, Property~\ref{W4} for $i=1,3$ is preserved from
  $(Q_1^{\cw},P_2^{\cw},S_3^{\cw})$ (by \cref{lem:leftAligned} for
  $i=3$).  Since $s\notin \V(P_3)\setminus \{f_3\}$, it follows that
  the endpoint $f_3$ of $P_{2}^{\cw}$ is not an interior vertex of the
  $st$-path on $S_{3}^{\cw}$.  Consequently, Property~\ref{W4} is
  also preserved for $i=2$.

  It remains to establish Property~\ref{W3}.  It is clearly
  preserved for $i=1$.  Consider the simple cycle $C'$ formed by
  $Q_{1}^{\cw}$, the $f_2f_3$-path on $\partial f$ that does not pass
  through $f_1$, $P_3$, and the $o_1o_3$-path on $o$ that does not
  pass through $o_2$.  By Obs.~\ref{obs:internally3conCycle}, the closed
  interior $C'^-$ of $C'$ is internally $3$-connected.  By
  construction, all vertices of~$Q_1^{\cw}$ and~$P_2^{\cw}$ belong to
  the closed exterior of~$C'$.  By planarity, there is no internal
  face of~$G$ that is not an internal face of~$C'^-$ and incident to
  more than one vertex of the $st$-subpath of $S_{3}^{\cw}$.  Hence,
  by \cref{lem:leftAligned} applied to $C'^-$ and the $st$-subpath of
  $S_{3}^{\cw}$, it follows that Property~\ref{W3} is preserved for
  $i=2$ and $i=3$.
\end{proof}

By Claim~\ref{claim:pathsQPT}, there is an archfree windmill, which
concludes the proof.
\end{proof}

\label{lem:strictlyInternalFace*}
\strictlyInternal*

\begin{proof}%
  Let~$n,m,f$ denote the number of vertices, edges, and faces of~$G$,
  respectively.  The existence of a separation pair would clearly
  violate Property~\ref{I3} of \cref{def:int3conDefs}.  Hence, the
  graph is 3-connected and it follows that each vertex
  of~$\partial o$ has degree 3 or 4.  The handshaking lemma implies
  that the number~$\eta$ of odd degree vertices is even.  Towards a
  contradiction, assume that~$G$ has no strictly internal face, i.e.,
  each face is incident to one of the vertices of~$\partial o$.  By
  $3$-connectivity, this implies that
  $f=7-\eta$.
  By Euler's polyhedron formula, the handshaking lemma, and internal
  4-regularity, we obtain:

\begin{eqnarray*}
  && n-m+f = 2
  \\
  &\Leftrightarrow& n - (2n-\frac{\eta}{2}) + (7-\eta) = 2
  \\
  &\Leftrightarrow& n=5-\frac{\eta}{2}
\end{eqnarray*}

If~$\eta=2$, it follows that~$n=4$, contradicting the internal
4-regularity.  If~$\eta=0$, then~$G$ is a 4-regular graph on 5
vertices, i.e., it is isomorphic to~$K_5$; a contradiction to the fact
that~$G$ is planar.
\end{proof}

\label{thm:fourRegAlgo*}
\fourRegAlgo*

\begin{proof}%
  The coordinates of the outer vertices of~$G$ are already fixed.  Our
  goal is to (recursively) compute coordinates for the internal
  vertices to obtain the desired drawing of~$G$.
  The base case of the recursion is that~$G$ contains no internal
  edges, in which case there is nothing to show.  So assume that~$G$
  has at least one internal edge.  Without loss of generality, we may
  assume that~$G$ contains no (outer) degree-2 vertices whose outer
  angle in~$\Gamma^o$ is~$\pi$ (we can just iteratively merge the two
  incident edges of such vertices, compute the drawing, and then
  reinsert the removed vertices at their prescribed coordinates).
  We distinguish two main cases:

\textbf{Case 1:} $G$ is not $3$-connected.
  We distinguish two subcases:

\textbf{Case 1.1:} $G$ contains a degree-2 vertex $v$. 
  For illustrations refer to \cref{fig:fourRegularAlgo1}a.  By
  internal $4$-regularity, $v$ belongs to~$\V(\Gamma^o)$.  By our
  assumption about degree-2 vertices, its outer angle in~$\Gamma^o$ is
  reflex.  Let~$u$ and~$w$ denote the two neighbors of~$v$.  Note that
  if~$uw\in E$, then it belongs to the boundary of the (triangular)
  internal face incident to~$v$ since otherwise~$u,w$ would form a
  separation pair that separates the interior of the cycle~$uvw$ from
  the outer face; contradicting Property~\ref{I3} of
  \cref{def:int3conDefs}.  If~$uw\in E$ we set~$G_1'=G$.  Otherwise,
  we add the edge~$uw$ in the internal face incident to~$v$ and call
  the resulting graph~$G_1'$.  Adding an internal edge to an
  internally $3$-connected graph clearly preserves
  Property~\ref{I2} of \cref{def:int3conDefs}, so, in both cases,
  $G'$ is internally $3$-connected.  We delete~$v$ from~$G_1'$ and
  call the resulting graph~$G_1$.  This modification preserves
  Property~\ref{I2} of \cref{def:int3conDefs}, so~$G_1$ is
  internally $3$-connected.

  Since~$\Gamma^o$ is compatible, the vertices~$u$ and~$w$ cannot
  belong to a common segment~$s$ of~$\Gamma^o$ (otherwise the internal
  face of~$G$ incident to~$v$ would arch~$s$).  Consequently, we can
  replace the edges~$uv$ and~$vw$ of~$\Gamma^o$ with the edge~$uw$ to
  obtain a {\em simple} convex polygon~$\Gamma^o_1$, which is a convex
  drawing of the outer face of~$G_1$.  By \cref{lem:minusTwo}, the
  edge~$uw$ is archfree in~$G_1$.  Combined with the fact
  that~$\Gamma^o$ is compatible with~$G$, it follows that~$\Gamma^o_1$
  is compatible with~$G_1$.

  We recursively compute the coordinates of the internal vertices in a
  convex drawing of~$G_1$ with~$\Gamma^o_1$ as the realization of the
  outer face such that all internal vertices of~$G_1$ are placed in
  the interior of some segment.  Since each internal vertex of~$G$ is
  also an internal vertex of~$G_1$, these coordinates combined with
  the coordinates of~$v$ correspond to the desired drawing of~$G$.

\begin{figure}[tbh]
  \centering \includegraphics[page=1]{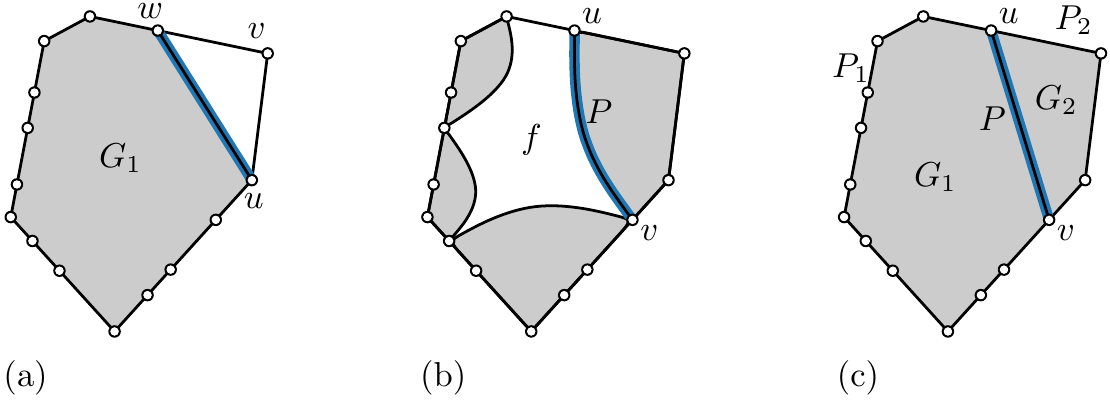}
  \caption{(a) Case 1.1 and (b,c)
    Case 1.2 in the proof of
    \cref{thm:fourRegAlgo}.  Note that in (b) the boundary of $f$
    contains only one internal $uv$-path.}
  \label{fig:fourRegularAlgo1}
\end{figure}

\textbf{Case 1.2:} $G$ contains no degree-2 vertex. 
  Let~$u$ be a vertex that belongs to a separation pair in~$G$.  By
  Property~\ref{I3} of \cref{def:int3conDefs}, all vertices that
  belong to separation pairs are on the boundary of the outer
  face~$o$ of~$G$.  Let~$v$ be the first vertex encountered when
  walking from~$u$ along~$\partial o$ in clockwise direction such
  that~$u,v$ is a separation pair, for an illustration see
  \cref{fig:fourRegularAlgo1}b.  By $2$-connectivity, there is an
  internal face~$f$ such that~$u,v\in\V(\partial f)$.  The
  boundary~$\partial f$ contains two simple interior disjoint
  $uv$-paths.  By the case assumption and the definition of~$v$, at
  least one of these paths, say~$P'$, is internal.  In fact, if
  $|\E(P')|= |\E(\partial f)|-1$, then the other path just consists of
  a single edge and is therefore also internal.  So in any case, the
  boundary~$\partial f$ contains an internal $uv$-path~$P$ with
  $|\E(P)|\le |\E(\partial f)|-2$, which, by \cref{lem:minusTwo}, is
  archfree.

  The boundary~$\partial o$ contains two interior disjoint
  paths~$P_1,P_2$ between~$u$ and~$v$.  Each of these two paths forms
  a simple cycle together with~$P$.  The closed interior of each of
  these two cycles describes an internally $3$-connected plane graph
  by Obs.~\ref{obs:internally3conCycle}.  We denote these two graphs
  by~$G_1$ and~$G_2$ such that~$G_i$, $i\in [2]$ has~$P_i$ on its
  outer face.  We define~$\Gamma^o_1$ to be the polygon resulting from
  replacing the part of~$\Gamma^o$ that corresponds to~$P_2$ with~$P$
  drawn as a straight-line segment, see \cref{fig:fourRegularAlgo1}c.
  The drawing~$\Gamma^o_2$ is defined analogously.  Since~$\Gamma^o$
  is compatible with~$G$, the vertices~$u,v$ cannot belong to a common
  segment~$s$ of~$\Gamma^o$ (otherwise, since there are no degree-2
  vertices (with outer angle $\pi$), $s$ would be arched by the
  internal face~$f$).  Hence, $\Gamma^o_1$ and~$\Gamma^o_2$ correspond
  to {\em simple} (convex) polygons.  Moreover, since~$\Gamma^o$ is
  compatible and~$P$ is archfree, $\Gamma^o_1$ and~$\Gamma^o_2$ are
  compatible with~$G_1$ and~$G_2$ respectively.  We recursively
  compute the coordinates of the internal vertices in convex drawings
  of~$G_1$ and~$G_2$ with outer face~$\Gamma_1^o$ and~$\Gamma_2^o$,
  respectively, where each internal vertex is placed in the interior
  of some segment.  Since, additionally, the interior vertices of~$P$
  are contained in the interior of the segment corresponding to~$P$,
  the combination of these drawings corresponds to the desired drawing
  of~$G$.

\textbf{Case 2:} $G$ is $3$-connected. 
  We distinguish two subcases:

\textbf{Case 2.1:} $|\V(\Gamma^o)|\ge 4$. 
    For illustrations refer to \cref{fig:fourRegularAlgo2}a.  Then
    there exist two distinct outer vertices~$u,v$ that do not belong
    to a common segment of~$\Gamma^o$.  By $3$-connectivity, $G$
    contains an internal $uv$-path~$P'$.  Consequently, by
    \cref{lem:leftAligned}, $G$ contains an archfree internal
    $uv$-path~$P$.  The boundary of the outer face of~$G$ contains two
    interior disjoint paths~$P_1,P_2$ between~$u$ and~$v$.  Each of
    these two paths forms a simple cycle together with~$P$.  The
    closed interior of each of these two cycles describes an
    internally $3$-connected plane graph by
    Obs.~\ref{obs:internally3conCycle}.  We denote these two graphs
    by~$G_1$ and~$G_2$ such that~$G_i$, $i\in [2]$ has~$P_i$ on
    its outer face.  We define~$\Gamma^o_1$ to be the polygon
    resulting from replacing the part of~$\Gamma^o$ that corresponds
    to~$P_2$ with~$P$ drawn as a straight-line segment.  The
    drawing~$\Gamma^o_2$ is defined analogously.  By definition of~$u$
    and~$v$, $\Gamma^o_1$ and~$\Gamma^o_2$ correspond to {\em simple}
    (convex) polygons.  More, since~$\Gamma^o$ is compatible and~$P$
    is archfree, $\Gamma^o_1$ and~$\Gamma^o_2$ are compatible
    with~$G_1$ and~$G_2$ respectively.  We recursively compute the
    coordinates of the internal vertices in convex drawings of~$G_1$
    and~$G_2$ with outer face~$\Gamma_1^o$ and~$\Gamma_2^o$,
    respectively, where each internal vertex is placed in the interior
    of some segment.  Since, additionally, the interior vertices
    of~$P$ are contained in the interior of the segment corresponding
    to~$P$, the combination of these drawings corresponds to the
    desired drawing of~$G$.

  \begin{figure}[tb]
    \centering
    \includegraphics[page=2]{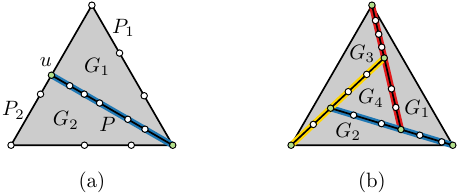}
    \caption{(a) Case 2.1 and (b)
      Case 2.2 in the proof of
      \cref{thm:fourRegAlgo}.}
    \label{fig:fourRegularAlgo2}
  \end{figure}

\textbf{Case 2.2:} $|\V(\Gamma^o)|= 3$. 
    By \cref{lem:strictlyInternalFace}, $G$ contains a strictly
    internal face.  Thus, by \cref{lem:windmill}, it contains an
    archfree windmill $(P,Q,S)$.  The paths~$P,Q,S$ dissect~$G$ into
    four plane graphs $G_1,G_2,G_3,G_4$, which are internally
    $3$-connected by Obs.~\ref{obs:internally3conCycle}.  The outer
    endpoints of~$P,Q,S$ correspond to the exactly three vertices
    of~$\Gamma^o$.  Consequently, they do not belong to a {\em common}
    segment of~$\Gamma^o$.  Hence, it is possible to draw each
    of~$P,Q,S$ as a straight-line segment such that the
    polygon~$\Gamma^o$ is dissected into four simple convex polygons,
    as depicted in \cref{fig:fourRegularAlgo2}b.  Each of these
    polygons corresponds to a convex drawing of the outer face of one
    of $G_1,G_2,G_3,G_4$.  Moreover, these drawings are compatible
    with their respective subgraphs since~$\Gamma^o$ is compatible and
    $P,S,Q$ are archfree.  We recursively draw $G_1,G_2,G_3,G_4$ into
    their respective compatible convex polygons in a convex fashion
    such that each of their internal vertices is placed in interior of
    some segment.  Since, additionally, all internal vertices of~$G$
    that belong to~$P,Q$, or~$S$ have been drawn in the interior of
    one of the segments corresponding to~$P,Q$, and~$S$, the
    combination of the four drawings corresponds to the desired
    drawing of~$G$.
\end{proof}

It is easy to see that the proof of \cref{thm:fourRegAlgo} corresponds
to a polynomial-time algorithm. In fact, it seems very plausible that
it can be implemented in quadratic time, though, we have not worked
out the details yet.

\label{thm:fourUnivUpper*}
\fourUnivUpper*

\begin{proof}%
  We create a convex drawing~$\Gamma^o$ of the outer face of~$G$ on
  exactly $3$ segments.  Let~$u,v,w$ be the three vertices
  of~$\Gamma^o$ whose outer angles are reflex.  By $3$-connectivity,
  none of the segments of~$\Gamma^o$ can correspond to a path that is
  arched by an internal face.  Consequently, $\Gamma^o$ is compatible
  with~$G$ and, by \cref{thm:fourRegAlgo}, we can create a convex
  drawing~$\Gamma$ of~$G$ that uses~$\Gamma^o$ as the outer face such
  that each vertex in~$\V(G)\setminus \{u,v,w\}$ is drawn in the
  interior of some segment of~$\Gamma$.  Hence, for each
  vertex~$x\in \V(G)\setminus \{u,v,w\}$ at most two segments
  of~$\Gamma$ have~$x$ as an endpoint.  For each
  vertex~$y\in \{u,v,w\}$ at most four segments of~$\Gamma$ have~$y$
  as an endpoint.  Since each segment has exactly two endpoints, it
  follows that the number of segments is at most
  $\frac{2(n-3)+4\cdot 3}{2}=n+3$, which concludes the proof.
\end{proof}

\label{prop:square*}
\propSquare*

\begin{proof}
  Suppose that, for $n \ge 6$, the graph $C_n^2$ has a
  drawing~$\Gamma$ with at most $n-1$ segments.  For
  $i \in \{0,2,4\}$, let $n_i$ be the number of vertices in $C_n^2$
  with $i$ ports.  Clearly, $n=n_0+n_2+n_4$.  The drawing~$\Gamma$ has
  $2n_2+4n_4$ ports and hence $n_2+2n_4$ segments.  If $n_4 \ge n_0$,
  then $\Gamma$ has $n_2+2n_4 \ge n_2+(n_0+n_4)=n$ segments, which
  would contradict our assumption.  Hence, $n_0 > n_4$.  Each vertex
  on the convex hull of~$C_n^2$ has four ports, which implies that
  $n_4 \ge 3$.  This in turn yields that $n_2=n-n_0-n_4 \le n-7$.

  We label the vertices of $C_n^2$ such that
  $\langle v_1,v_2,\dots,v_n \rangle$ forms the simple cycle~$C_n$.
  Since $n_0>n_4$, there must be two indices $1 \le j<\ell \le n$ such
  that vertices~$v_j$ and~$v_\ell$ have zero ports and every
  vertex~$v_k$ with $j<k<\ell$ has two ports.  Let
  $V'=N(v_j) \cup \{v_j,\dots,v_\ell\} \cup N(v_\ell)$ and $n':=|V'|$.
  We have $n' \le n-1$ since $V'$ contains at most $n-7$ vertices
  (with two ports) strictly between~$v_j$ and~$v_\ell$, plus six
  further vertices.  Hence, w.l.o.g., we can choose our labeling of~$C_n^2$
  such that $V'=\{v_{j-2},\dots,v_{\ell+2}\}$.  Let $G'=G[V']$, but
  drop any edge that connects one of the first two with one of the
  last two vertices.  Then $G'$ is isomorphic to the outerpath
  $R_{n'}$ where every vertex has degree at most~4.  Dujmovi\'c et
  al.\ have shown that $\seg(R_{n'}) = n'$~\cite[Proof of Theorem 7]{desw-dpgfss-CGTA07}.
  The graph $G'$, however, has only $2n'-2$ ports: all vertices have
  two ports, except for $v_j$ and $v_\ell$ with zero ports and
  $v_{j-1}$ and $v_{\ell+1}$ (both of degree~3) with at most three
  ports.  This contradicts the fact that $\seg(R_{n'})=n'$.
\end{proof}

\begin{remark}
  \label{rem:triconn-universal}
  As a universal lower bound for the class of 3-connected 4-regular
  planar graphs, note that in each vertex either at least two segments
  end or two segments cross.  In order to generate $n$ vertices, we
  need at least $\Omega(\sqrt{n})$ segments as Dujmovi\'c et
  al.~\cite{desw-dpgfss-CGTA07} observed.  It is not hard to see
  that (grid-like) 3-connected 4-regular planar graphs with segment
  number $O(\sqrt{n})$ exist.
\end{remark}

\section{Proofs Omitted in Section \ref{sec:outerpaths} (Maximal
  Outerpaths)}
\label{app:outerpaths}

\label{lem:long*}
\activelong*

\begin{proof}
  Suppose that $\Gamma_i$ contains two pseudo-$k$-arcs $\alpha$ and $\beta$
  that are both active and long.  Let $\alpha$ have its first internal
  edge before $\beta$.  For $\beta$ to become long, $\beta$ must have
  $k+1$ internal edges, while $\alpha$ remains active.  Let
  $H_0, H_1, \dots, H_{\ell}$ be the subgraphs into which the internal
  edges of $\beta$ subdivide the complete outerpath drawing $\Gamma$; see
  \cref{fig:definitionOfHs}.  Now for $\alpha$ to leave $H_0$,
  $\alpha$ needs either to enter $H_1$ (which requires an intersection
  between $\alpha$ and $\beta$) or to enter $H_2$ (which requires a
  tangential point of $\alpha$ at $\beta$ and is counted as two
  intersections).
  For $\alpha$ to be active when $\beta$ is long, $\alpha$ needs to
  reach $H_{k + 1}$ (or some $H_j$ with $j>k+1$).  This however,
  requires at least $k + 1$ intersection points between $\alpha$ and
  $\beta$, a contradiction to the definition of pseudo-$k$-arcs.
\end{proof}

\label{clm:transition-loss-k=2*}
\transitionlossarcs*

\begin{proof}%
  Of course, the loss cannot be negative and the number of transitions
  from one long arc to the other is $\arc_k^{>k} - 1$.

  Summing up the losses over all pseudo-$k$-arcs of the drawing, we
  obtain $t_k$.  Being counted in a crossing with a long arc more than
  $k$ times is no contradiction to the definition of pseudo-$k$-arcs
  because the long arc may change.  We distinguish two cases for the
  transition of a long arc $\alpha$ (with internal edge
  $e_{1}, \dots, e_{\ell}$ and subgraphs $H_1, \dots, H_{\ell}$) to a
  long arc $\beta$ (with internal edge $e'_{1}, e'_{2}, \dots$ and
  subgraphs $H'_1, H'_2, \dots$).
	
  \begin{figure}[tbh]
    \centering
    \includegraphics{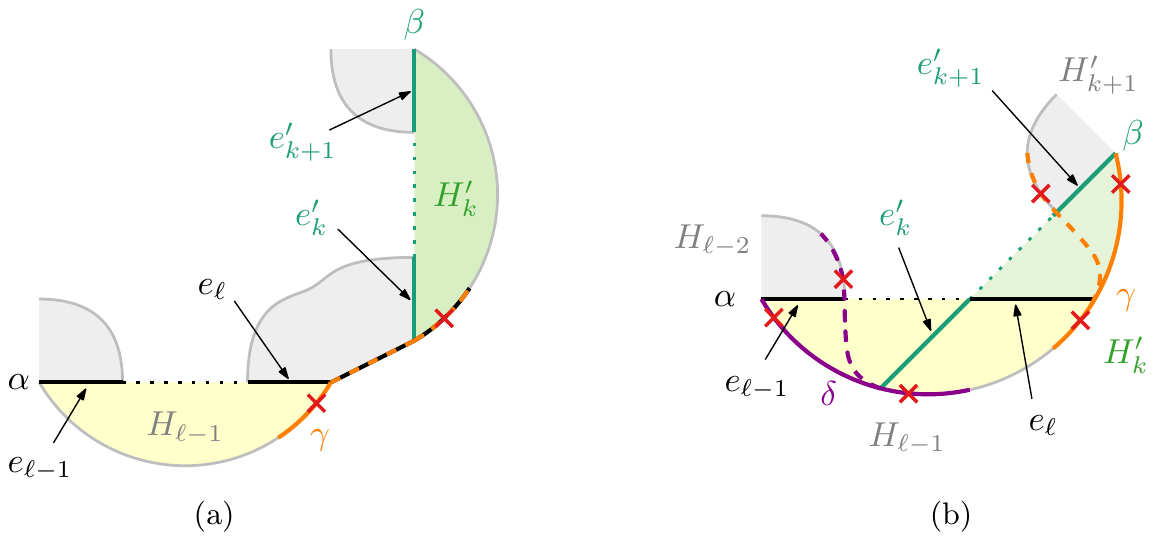}
    \caption{Cases for the transition of one long pseudo
      segment~$\alpha$ to another long pseudo segment~$\beta$.}
    \label{fig:transitionK1}
  \end{figure}
	
  In the \textbf{first case}, $e_{\ell}$ precedes $e'_{k}$; see
  \cref{fig:transitionK1}a.  Say an arc $\gamma$ has been counted in
  $q$ crossings with a long arc before reaching $e'_{k}$. (If there
  has been a transition of a long arc before, we have already
  subtracted its loss and hence we assume $q \le k$.)  Next we show
  that $\gamma$ is part of at most $k - q + 1$ counted crossings with
  $\beta$.  When $\gamma$ reaches $e'_{k}$, it must have intersected
  $\beta$ already at least $q - 1$ times.  This is due to the fact
  that $k - 1 \ge q - 1$ internal edges of $\beta$ precede $e'_{k}$
  and while an arc is active, it intersects all arcs of the internal
  edges.
  We know that $\gamma$ has intersected $\alpha$ (or a previous long
  arc) $q$ times, so it has been in at least the last $(q - 1)$ bays
  of $\alpha$ (or a previous long arc).  By then, $\beta$ has also
  already been active and also has been in at least the last $(q - 1)$
  bays of $\alpha$ (or a previous long arc).  In any bay~$H$, all arcs
  that leave $H$ intersect all other arcs that leave $H$ at least
  once.
  Hence, $\beta$ and $\gamma$ have intersected pairwise at least
  $q - 1$ times.  This means that $\gamma$ can be part of at most
  $q + (k - q + 1) = k + 1$ counted crossings with a long arc~--
  regarding all long arcs up to and including $\beta$.  It remains to
  argue that there is at most one arc $\gamma$ with $k + 1$ counted
  crossings with a long arc per transition.  Suppose there was another
  arc $\gamma'$ with the same property, which has been in $q'$
  crossings with long arcs before $e'_{k}$.  Then, without loss of
  generality, $\gamma$ has been in the crossing with $\alpha$ at
  $e_{\ell}$.  However, $\gamma'$ has also intersected $\alpha$ at
  $e_{\ell}$ but without being counted in a crossing.  So, $\gamma'$
  has been in the last $q'$ $H$s of $\alpha$ together with $\beta$ and
  contributes at most $k - q'$ crossings with $\beta$.
	
  In the \textbf{second case}, $e_{\ell}$ succeeds $e'_{k}$; see
  \cref{fig:transitionK1}b.  If we started counting crossings
  with~$\beta$ in bay $H'_{k+1}$ instead of $H'_{k}$, we would have
  the same situation as in the first case.  Now consider the counted
  crossing of $H'_k$ at $e'_{k + 1}$.  Similar to the first case, if
  an arc $\delta$ reaches this crossing and was part of counted
  crossings before, it has intersected $\alpha$ at $e_{\ell}$. Again,
  only one of $\gamma$ and some other pseudo-$k$-arc $\gamma'$ can
  contribute the crossing with $\alpha$ at $e_{\ell}$ and then be part
  of $k + 1$ crossings with long edges.  For the counted crossing of
  bay $H'_k$ at $e'_{k}$, we cannot rule out the possibility that the
  involved arc~$\delta$ is part of more than~$k$ crossings.  So, we
  consider this crossing as being lost, but then there exists a
  crossing of $\alpha$ and $\beta$ that has not been counted~-- namely
  at the common vertex of $e_{\ell}$ and $e'_{k}$.  Therefore, also in
  the second case we have a loss of at most one counted crossing.
\end{proof}

\label{clm:outerpath-has-2-zero-segs-and-3-one-segs*}
\outerpathZeroOneSegs*

\begin{proof}%
  Consider $v_1$ and $v_n$, i.e., the first and the last vertex in the
  stacking order of $G$.  Each of them is incident to two pseudo
  segments.  If they would lie on only one pseudo segment $S$, $S$
  would intersect the pseudo segment connecting the two neighbors of
  $v_1$ (or $v_n$) twice.
	
  First, we show that $v_1$ and $v_n$ have at least one incident
  pseudo segment with zero internal edges each (\textbf{Case 0}).  Without loss of
  generality, assume that $v_1$ is incident to the pseudo segments
  $S_l$ and $S_r$, both have at least one internal edge, and in the
  stacking order of the outerpath, the first internal edge~$e$
  of~$S_r$ precedes the first internal edge of $S_l$; see
  \cref{fig:pseudo-segments-with-few-internal-edges}a.  The path of
  faces reaches the face~$f$ when passing over $e$.  However, $S_l$ is
  not incident to $f$ and becomes inactive.  ($S_l$~cannot be incident
  to~$f$ because then $S_l$ and $S_r$ would intersect twice or $v_1$
  would have a degree $>2$.)  Therefore, $S_l$ has zero internal
  edges.  The same holds when traversing the outerpath backwards
  starting at~$v_n$.

  \begin{figure}[tbh]
	\centering \includegraphics{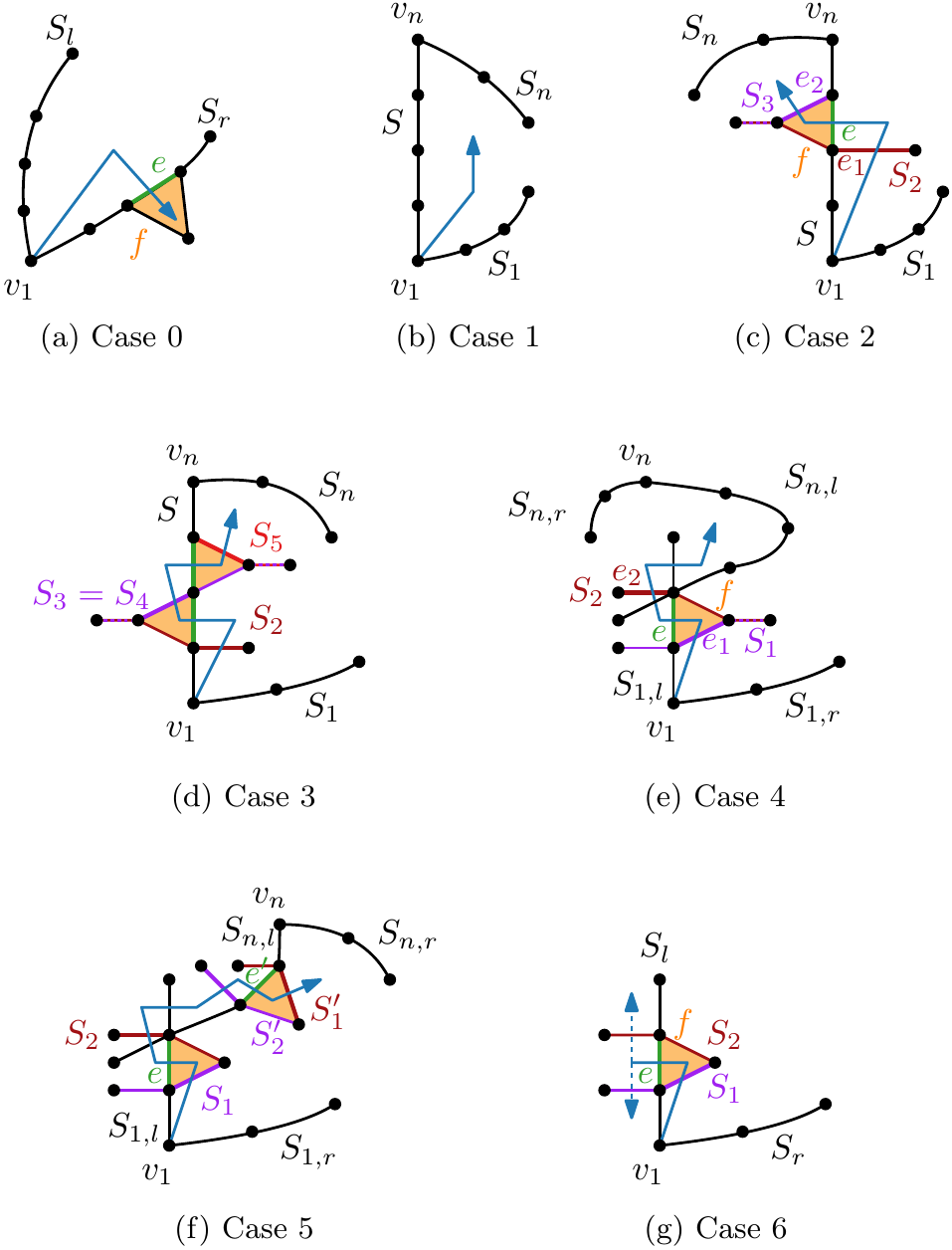}
    \caption{Cases to show
      \cref{clm:outerpath-has-2-zero-segs-and-3-one-segs}.  }
    \label{fig:pseudo-segments-with-few-internal-edges}
  \end{figure}

  Using this property, we now can make the following case distinction.
	
  \textbf{Case 1:} $v_1$ and $v_n$ are incident to the same pseudo
  segment~$S$ having zero internal edges.  Let the other pseudo
  segments being incident to $v_1$ and $v_n$ be $S_1$ and $S_n$,
  respectively (clearly, they are distinct); see
  \cref{fig:pseudo-segments-with-few-internal-edges}(b).  This means
  that $S$ is incident to all faces in the outerpath.  So, if $S_1$ or
  $S_n$ had an internal edge, they would intersect $S$ a second
  time. Hence, $S_1$ and $S_n$ have also zero internal edges and we
  have at least three pseudo segments with zero internal edges in
  total.
	
  \textbf{Case 2:} $v_1$ and $v_n$ are incident to the same pseudo
  segment~$S$ having one internal edge.  Let the other pseudo segments
  being incident to $v_1$ and $v_n$ be $S_1$ and $S_n$, respectively
  (clearly, they are distinct); see
  \cref{fig:pseudo-segments-with-few-internal-edges}c.  Since $S$ has
  an internal edge, $S_1$ and $S_n$ have zero internal edges.
  Consider the face~$f$ following the internal edge~$e$ of~$S$.
  Beside $S$, the two other distinct bounding pseudo segments of~$f$
  are~$S_2$ and~$S_3$.  Let $S_3$ have an internal edge~$e_2$
  following $e$ along the sequence of internal faces.  All faces of
  the outerpath are incident to $S$, hence $S_3$ cannot have a second
  internal edge as it intersects $S$ incident to $f$.  Similarly,
  $S_2$ can have at most one internal edge~$e_1$ when it intersects
  $S$ incident to $f$.  Thus, $S_1$ and $S_n$ have zero internal
  edges, while $S$, $S_2$, and $S_3$ have at most one internal edge
  each.
	
  \textbf{Case 3:} $v_1$ and $v_n$ are incident to the same pseudo
  segment~$S$ having at least two internal edges.  As in Case~2, when
  the sequence of faces of the outerpath passes over $S$, there are
  two pseudo segments $S_2$ and $S_3$ each having at most one internal
  edge.  We have this situation at least twice~-- we denote the next
  corresponding pair of segments that has at most one internal edge
  per pseudo segment by $S_4$ and~$S_5$.  Observe that maybe
  $S_3 = S_4$; see
  \cref{fig:pseudo-segments-with-few-internal-edges}d.  Then, however,
  $S_2 \ne S_5$ as otherwise $S_2$ and $S_3$ would intersect twice.
  Therefore, we have two pseudo segments with zero internal edges
  ($S_1$ and $S_n$) and we have at least three pseudo segments with at
  most one internal edge ($S_2$, $S_3$, and $S_5$).
	
  \textbf{Case 4:} $v_1$ and $v_n$ are incident to four distinct
  pseudo segment and exactly one of these pseudo segments has at least
  two internal edges.  This case is similar to Case~2 and Case~3.
  Without loss of generality, let the segments of~$v_1$ and $v_n$ be
  $S_{1, l}$, $S_{1, r}$ and $S_{n, l}$, $S_{n, r}$, respectively, and
  let $S_{1, l}$ have at least two internal edges; see
  \cref{fig:pseudo-segments-with-few-internal-edges}e.  Consider the
  first internal edge~$e$ of~$S_{1, l}$ and the face~$f$ preceding~$e$
  along the sequence of faces.  Let the other pseudo segments bounding
  $f$ be $S_1$ and $S_2$ and let the internal edge~$e_1$ for entering
  $f$ be contained in $S_1$.  Because until the sequence of faces
  passes over $S_{1, l}$ a second time, all faces are neighboring
  $S_{1, l}$.  Hence, $S_1$ and $S_2$ have at most one internal edge
  each.  Moreover, observe that neither $S_{n, l}$ nor $S_{n, r}$ can
  be equal to $S_1$ or $S_2$ as otherwise they would intersect
  $S_{1, l}$ twice.  This gives us our bound~-- the three pseudo
  segments with at most one internal edge are $S_1$, $S_2$, and one of
  $S_{n, l}$ and $S_{n, r}$.
	
  \textbf{Case 5:} $v_1$ and $v_n$ are incident to four distinct
  pseudo segment and two of these pseudo segments have at least two
  internal edges.  We have a very similar situation as in Case~$4$,
  but now we have $S_1$ and $S_2$ in the forward direction and $S_1'$
  and $S_2'$ symmetrically in the backward direction; see
  \cref{fig:pseudo-segments-with-few-internal-edges}f.  Let $S_{1, l}$
  and $S_{n, l}$ be the segments originating at $v_1$ and $v_n$,
  respectively, that have at least two internal edges each.  We have
  to be a bit more careful about the case that $S_{1, l}$ and
  $S_{n, l}$ intersect.  However, even in this case $S_1$, $S_2$,
  $S_1'$, and $S_2'$ are four distinct pseudo segments since the first
  internal edge~$e$ of $S_{1, l}$ precedes all internal edges of
  $S_{n, l}$ and the last internal edge~$e'$ of $S_{n, l}$ succeeds
  all internal edges of $S_{1, l}$ (otherwise the drawing would not be
  an outerpath).
	
  \textbf{Case 6:} $v_1$ and $v_n$ are incident to four distinct
  pseudo segment and each of them has at most one internal edge.  If
  three of them have zero internal edges, we are done.  So assume that
  the pseudo segment $S_l$ of $v_1$ (and one pseudo segment of $v_n$)
  has one internal edge~$e$; see
  \cref{fig:pseudo-segments-with-few-internal-edges}g.  Consider the
  face~$f$ preceding $e$ in the sequence of faces in the outerpath.
  Beside $S_l$, let $f$ be bounded by $S_1$ and $S_2$.  The key
  insight is that $S_1$ and $S_2$ pass over $S_l$ at $e$, but on the
  other side of $S_l$, they cannot intersect a second time and so the
  path of faces in the outerpath can yield another internal edge at
  most for one of $S_1$ and $S_2$.  Hence, either $S_1$ has at most
  one internal edge (when entering $f$) or $S_2$ has zero internal
  edges, which provides our bound.  We have to be careful about the
  case that $S_1$ or $S_2$ are pseudo segments of~$v_n$.  Note that
  not both of them can reach $v_n$ because then they would intersect a
  second time. If $S_2$ reaches $v_n$, then $S_1$ is our third pseudo
  segment with at most one internal edge.  If $S_1$ reaches~$v_n$,
  then $S_2$ is our third pseudo segment without any internal edges.
\end{proof}

\label{clm:outerpath-segs*}
\outerpathSegs*

\begin{proof}
	Clearly, $\seg(G) \ge \arc_1(G)$.
	Hence, it suffices to show $\arc_1(G) \ge \floor{\frac{n}{2}} + 2$.
	
	We plug in the result from \cref{clm:transition-loss-k=2}, into
	\cref{eq:general-formula-pseudo-k-arcs2} for $k=1$ and use
	\cref{clm:outerpath-has-2-zero-segs-and-3-one-segs} to observe
	$3\arc_1^0 + \arc_1^1\ge 9$:
	\begin{align*}
	\arc_1 &\ge \frac{2n - 3 + 3\arc_1^0 + \arc_1^1}{4}
	= \frac{n}{2} + \frac{3\arc_1^0 + \arc_1^1-3}{4} \ge \frac{n+3}{2}
	\end{align*}
	As we cannot have partial (pseudo) segments, we can round up to
	$\ceil{\frac{n+3}{2}}=\floor{\frac{n}{2}} + 2$.
\end{proof}

\label{clm:outerpath-arcs*}
\outerpathArcs*

\begin{proof}
	Clearly, $\arc(G) \ge \arc_2(G)$.
	Hence, it suffices to show $\arc_2(G) \ge \ceil{\frac{2n}{7}}$.
	
	For $k = 2$ and \cref{eq:general-formula-pseudo-k-arcs2},
	we plug in the result from \cref{clm:transition-loss-k=2} and we get
	
	$$\arc_2 \ge \frac{2n + 5\arc_2^0 + 3\arc_2^1 + 1\arc_2^2}{7}
	\ge \frac{2n}{7} \, .$$
	Since we can only have an integral number of arcs,
	we can round up this value.
\end{proof}

\label{clm:outerpathExamples*}
\outerpathExamples*

\begin{proof}
  (i) Consider~\cref{fig:tightOuterpaths}a.  In the base case
  ($m = 0$), there obviously is a drawing with six vertices on five
  line segments.  When we increase $m$ by one, we add a line segment
  going through the central vertex and increasing the number of
  vertices by two.

  (ii) Consider~\cref{fig:tightOuterpaths}b, where $m = 6$.  The main
  structure is a long horizontal line segment (this is a circular arc
  with radius $\infty$).  In the base case ($m = 2$), we have two more
  circular arcs that look like the first and the last arc
  in~\cref{fig:tightOuterpaths}b~-- two of their vertices are shared
  with each other, which gives us 6 vertices in total.  When we
  increase $m$ by one, we add a circular arc as in
  \cref{fig:tightOuterpaths}b.  It has six vertices, where three of
  them are new.
	
  (iii) Consider~\cref{fig:tightOuterpaths}c.  In the base case
  ($k = 0$), we have only the first three and the last three vertices
  (in purple) using three pseudo-2-arcs.  When we increase $k$ by one,
  we add the colored part $k$ times (to show the repeating pattern,
  there is another copy in gray).  This colored part has 16 vertices,
  it extends three pseudo-2-arcs and introduces five new pseudo
  2-arcs.  Observe that each pair of pseudo-2-arcs intersects at most
  twice.
\end{proof}

\section{Maximal Outerplanar Graphs and 2-Trees}
\label{app:outerplanar}

Consider a straight-line drawing $\Gamma_G$ of a 2-tree $G$.  The main
idea for a universal lower bound for 2-trees
(and for its subclass of maximal outerplanar graphs)
is that~$G$ either has many degree-2 vertices and thus
requires many segments
(recall that, in a 2-tree, all faces are triangles, hence
degree-2 vertices cannot be closed)
or $G$ can be obtained by gluing few
outerpaths for which we know (tight) universal lower bounds
on the segment number.
By gluing we mean the following.  Let $G$ be a 2-tree and $P$ a
maximal outerpath.  Let $f_G$ be a triangle of $G$ that is not
incident to a degree-2 vertex and let $f_P$ be a triangle of~$P$ that is
incident to a degree-2 vertex (i.e., $f_P$ is the first or last
triangle of~$P$).
Let $\Gamma_P$ be a straight-line
drawing of $P$.  Then we define the \emph{gluing} of $\Gamma_P$ to
$\Gamma_G$ as the straight-line drawing $\Gamma_{G \oplus P}$ of the
2-tree $G \oplus P$ obtained by identifying $f_P$ and $f_G$; see
\cref{fig:gluing}.  Note that
$\abs{V(G \oplus P)} = \abs{V(G)} + \abs{V(P)} - 3$.
In~$\Gamma_{G \oplus P}$, we call $f_G$
and $f_P$ the \emph{gluing faces} of $G$ and $P$, respectively.

\begin{figure}[tbh]
  \centering
  \includegraphics{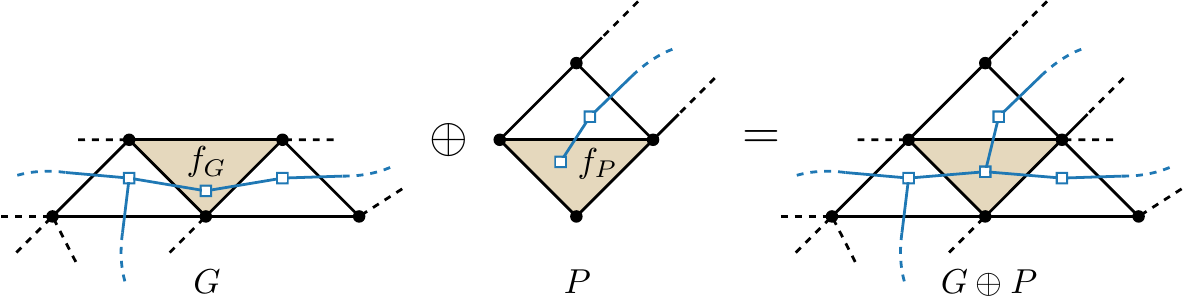}
  \caption{Gluing drawings of an outerplanar graph $G$ and an outerpath $P$.}
  \label{fig:gluing}
\end{figure}

Unfortunately, for gluing outerpaths,
we cannot directly employ \cref{clm:outerpath-segs}
because it does not tell us how many ports we lose when gluing.
Therefore, we first investigate the distribution of ports
within a straight-line drawing of a maximal outerpath.
We will see that, by some careful counting arguments, we
lose only few (counted) ports when gluing outerpaths.
We start by formally proving some auxiliary properties; see
\cref{fig:segOuterpathProperties}.

\begin{figure}[tbh]
	\centering \includegraphics[page=4]{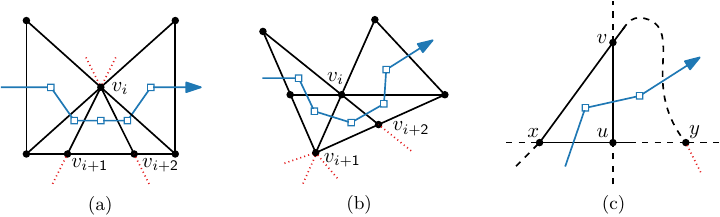}
	\caption{For property \ref{case:pSix} in
		\cref{clm:openClosedPathVertices}, observe that vertex~$v_i$
		(which has degree~6) is (a) either open or (b) has a bend companion
		(here $v_{i+1}$) with three ports;
		(c) for property~\ref{case:pFours}, note that two
		subsequent degree-4 vertices $u$ and $v$ cannot both be closed
		because of the two triangles they form with their common neighbors
		$x$ and $y$.}
	\label{fig:segOuterpathProperties}
\end{figure}

\begin{lemma}
	\label{clm:openClosedPathVertices}
	Let $P$ be a maximal outerpath given with a stacking order,
	and let~$v$ be a vertex of~$P$.
	Then, in any outerplanar straight-line drawing of $P$, all of the
	following holds.
	\begin{enumerate}[label=(P\arabic*),left=0pt,nosep]
		\item \label{case:p2odd} If $\deg(v) = 2$ or $\deg(v)$ is odd, then
		$v$ is open.
		\item \label{case:pFive} If $\deg(v) \geq 5$, then $v$ is succeeded
		by $\deg(v)-4$ many neighbors of degree~$3$,
		which we call \emph{companions}.
		\item \label{case:pSix} If $\deg(v) \geq 6$ and $v$ is closed, then
		$v$ has a companion with three ports,
		which we call \emph{bend companion}.
		\item \label{case:pFours} If subsequent vertices $u$ and $v$ both
		have degree $4$, then at least one of $u$ and $v$ is open.
		\item \label{case:pFourFive} Let $v$ be stacked upon the edge $uw$
		and $u, v$ be subsequent vertices.  If $v$ is closed,
		$\deg(v) = 4$, $\deg(u) = 3$, and $\deg(w) = 5$, then either $u$
		or $w$ has at least three ports.
	\end{enumerate}
\end{lemma}

\begin{proof}
	We consider each of the statements individually.
	\begin{enumerate}[label=(P\arabic*),left=0pt,nosep]
		\item If $\deg(v)$ is odd, the claim is trivial.  Otherwise, $v$ and
		its two neighbors form a triangle in any 2-tree and cannot be
		collinear.
		
		\item Since $v$ has degree at least five, constructing $P$ with a
		sequence of stacking operations involves $\deg(v)-2$ consecutive
		stacking operation on edges incident to $v$.  Consequently, all
		succeeding neighbors of $v$, except for the last two, must have
		degree three.
		\item Let $v = v_i$.  Consider the companions
		$v_{i+1}, \dots, v_{i+\deg(v_i)-4}$ of~$v$.  Suppose neither of
		them has three ports (two is not possible since they have
		degree~$3$), then (at least) $\deg(v_i) - 2$ neighbors of $v_i$
		are collinear and thus result in a triangle~$T$ with $v$ at one
		corner and these neighbors on the opposing side of~$T$; see
		\cref{fig:p6collinear}.  Then, however, $v_i$ cannot be closed
		since $\deg(v_i)-2 > \deg(v_i)/2$ and at most two segments can
		pass through~$v_i$, which is a contradiction.  Hence, one of the
		companion vertices of $v$ has three ports; see
		\Cref{fig:p6closed}.
		
		\item Let $x, u, v, y$ be a stacking subsequence in $P$ where both
		$u$ and $v$ have degree four.  Assume, for the sake of
		contradiction, that there exists a planar straight-line drawing of
		$P$ where both $u$ and $v$ are closed.  Let $s$ be the segment
		that contains the edge~$uv$.  Then $s$ intersects at $u$ the
		segment $s'$ that contains $xu$, and $s$ intersects at $v$ the
		segment $s''$ that contains $xv$; see \cref{fig:subsequent4s}.
		Observe that $s'$ and $s''$ need to intersect again in $y$ since
		$P$ is a maximal outerpath and both $u$ and $v$ are closed.
		However, this would only be possible if $x$, $u$, $v$, and $y$ are
		collinear; which is a contradiction to the drawing being a planar
		straight-line drawing.
		
		\item If $u$ has three ports, we are done.  Otherwise,
		$u$ has only one port; see \cref{fig:pFourFive}.
		Then note that $u$ and $v$ need to be collinear with a successor
		$v'$ of $v$ and predecessor $u'$ of $u$.  Observe that these four
		vertices are adjacent to $w$.  However, since $w$ has degree~5,
		only one of the edges
		$\set{u', w}, \set{u, w}, \set{v, w}, \set{v', w}$ can be extended
		at $w$.  (In \cref{fig:pFourFive}, the edge $\set{v', w}$ lies on
		a segment passing through $w$.)  Therefore, $w$ has at
                least three ports.
	\end{enumerate}
        This finishes the proof.
\end{proof}

\begin{figure}[tb]
	\centering
	\begin{subfigure}[t]{.45 \linewidth}
		\centering \includegraphics[page=1]{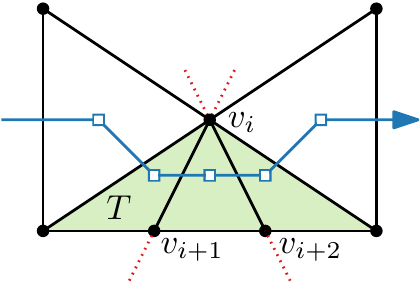}
		\caption{For \ref{case:pSix}, if all companions of $v$ are
			collinear, then $v$ cannot be closed.}
		\label{fig:p6collinear}
	\end{subfigure}
	\hfill
	\begin{subfigure}[t]{.45 \linewidth}
		\centering \includegraphics[page=2]{segOuterpathHighdegree3chains}
		\caption{For \ref{case:pSix}, if $v$ is closed, then one of its
			companion neighbors has 3 ports.}
		\label{fig:p6closed}
	\end{subfigure}
	
	\caption{Configurations in the proof of
		\cref{clm:openClosedPathVertices}, where $v$ has degree 6 (or a
		higher even degree).}
	\label{fig:highdegree3chains}
\end{figure}

\begin{figure}[tb]
	\centering
	\begin{subfigure}[t]{.45 \linewidth}
		\centering \includegraphics[page=1]{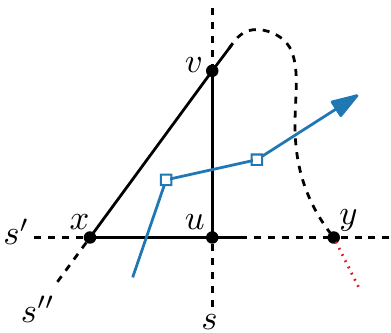}
		\caption{For \ref{case:pFours}, two subsequent degree-4
			vertices $u$ and $v$ cannot both be closed because of the two
			triangles with their common neighbors $x$ and $y$.}
		\label{fig:subsequent4s}
	\end{subfigure}
	\hfill
	\begin{subfigure}[t]{.45 \linewidth}
		\centering
		\includegraphics[page=1]{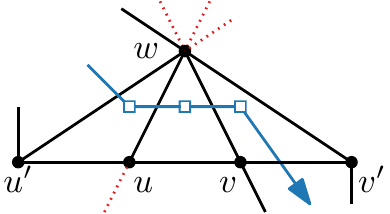}
		\caption{For \ref{case:pFourFive}, in case $u$ has only one
			port, $w$ has at least three ports because four of its neighbors
			are collinear.}
		\label{fig:pFourFive}
	\end{subfigure}
	\caption{Configurations in the proof of
		\cref{clm:openClosedPathVertices} where $v$ is closed and has
		degree 4.}
\end{figure}

\begin{proposition}
	\label{clm:outerpathSec}
	Let $P$ be a maximal outerpath with $n \geq 4$ vertices. Then
	$\openseg(P) \ge n + 1$.
	Moreover, for any planar straight-line drawing of~$P$,
	we can find an injective assignment of ports to vertices
	such that every port is assigned to its own vertex or to a
        neighboring vertex.
\end{proposition}

\begin{proof}
	Given any straight-line drawing~$\Gamma_P$ and any stacking
        order $\langle v_1,\dots,v_n \rangle$ of~$P$,
	we describe an assignment of ports to vertices in their vicinities such
	that no two ports are assigned to the same vertex.
	This immediately proves that $\openseg(P) \ge n$.
	For the one remaining port,
	observe that $v_1$ has an additional unassigned port.
	
	Let $i \in [n]$.
        We consider different situations for vertex~$v_i$.  Each situation
	is illustrated by a vertex in the example shown in
	\cref{fig:pathNess}.  If $v_i$ is open (such as $v_1$, $v_2$, or
	$v_4$ in \cref{fig:pathNess}), we assign one of the ports to itself.
	If $v_i$ is closed, then $\deg(v_i)$ is even and at least $4$ by
	\ref{case:p2odd}.  First assume $\deg(v_i) \geq 6$ (such as
	$v_9$, $v_{13}$, $v_{16}$ in \cref{fig:pathNess}).
        Then, by \ref{case:pSix}, we know
	that~$v_i$ has a bend companion~$v_j$ with three ports.  Only one of
	the three ports of~$v_j$ is assigned to~$v_j$ itself, so we assign
        one of the remaining ports of~$v_j$
	to~$v_i$.  (In \cref{fig:pathNess} such a port would be
        supplied by~$v_{11}$, $v_{14}$ and~$v_{18}$, respectively.)
	
	If $\deg(v_i) = 4$, then either $\deg(v_{i-1}) = 4$ (such as $v_4$
	preceding $v_5$) or $\deg(v_{i-1}) = 3$ (such as $v_{11}$ preceding
	$v_{12}$)
	since, by \ref{case:pFive}, $v_{i-1}$ has degree at most~4.
	In the former case,~$v_{i-1}$ has at least two ports by
	\ref{case:pFours} and we can assign one of the ports to $v_i$
	(such as $v_4$ to~$v_5$).  In the latter case,
	we distinguish three subcases.
	If $v_i = v_3$ in the stacking order of~$P$,
        then $v_2$ has degree~3 and cannot be
	closed (as~$v_2$ and~$v_3$ in \cref{fig:pathNess}).
	If $v_i = v_4$ in the stacking order of~$P$,
	then $v_2$ or $v_3$ has three ports.
	Otherwise, observe that	the common neighboring predecessor
	of $v_{i-1}$ and $v_i$ has degree at least~5; hence
	one of \ref{case:pSix} or \ref{case:pFourFive} applies
	(see $v_9$, $v_{11}$, and $v_{12}$ in \cref{fig:pathNess};
        this is the only case where
	a vertex provides ports for itself and two other vertices).
\end{proof}

\begin{figure}[tb]
	\centering \includegraphics{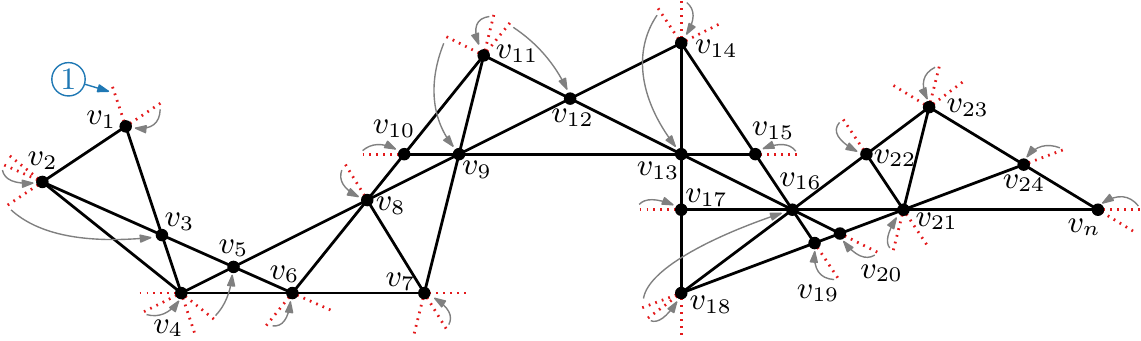}
	\caption{A straight-line drawing of a maximal outerpath where each
		vertex is assigned a port (indicated by grey arrows).  Several
		ports remain unassigned (e.g.,
                \raisebox{.5pt}{\textcircled{\raisebox{-.9pt} {1}}}).}
	\label{fig:pathNess}
\end{figure}

\Cref{clm:outerpathSec} implies a universal lower bound of
$(n+1)/2$ for the segment number of an $n$-vertex outerpath.
In \cref{clm:outerpath-segs}, we improve this by a constant.

\label{clm:outerplanarSec*}
\outerplanarSec

\begin{proof}
	For now, assume that $G$ is a maximal outerplanar graph.
	We consider the case that $G$ is a 2-tree at the end of this proof.
	If the weak dual $T$
	of $G$ has at least $(n+7)/5$ leaves, we are done since $G$ has at
	least as many segments as $T$ has leaves.
	
	Otherwise, let
	$\cP = \set{P_1, \ldots, P_p}$ be a minimum-size set of maximal
	outerpaths such that when we define $G_1 = P_1$ and
	$G_i = G_{i-1} \oplus P_i$, for $i \in \set{2, \ldots p}$, we get
	that $G = G_p$.  In other words, we can obtain $G$ by $p-1$
	consecutive gluing operations of the paths in $\cP$.
	Note that $p$ is at most $(n+7)/5 - 1 = (n+2)/5$
	because $P_1$ contains two leaves and, for $i \in \{2, \dots
        p\}$, $P_i$ contains one leaf of~$T$.

	Next, we show a lower bound
	on the number of ports on any straight-line drawing of $G$.  To this
	end, we use the assignment of ports to vertices that we established
	in \cref{clm:outerpathSec} and apply it to each outerpath~$P_i$ in~$\cP$.
	Further, we use the stacking order of~$P_i$
	that starts at the degree-2 vertex of~$P_i$ that is not incident to
	the gluing face of~$P_i$.  For
	$i \in \set{1, \dots, p}$, let $n_i = \abs{V(P_i)}$.  Note that
	$\abs{V(G)} = \sum_{i=1}^p n_i - 3(p-1)$.

	First, we compute $\openseg(\cP)$, the sum of ports counted for
	$P_1, \ldots, P_p$:
	\begin{align*}
	\openseg(\cP) & = \sum_{i=1}^p \openseg(P_i) \geq \sum_{i=1}^{p} (n_i+1)
	= n + 3(p-1) + p
	= n + 4p - 3
	\end{align*}
	
	Second, we analyze the number of {\em counted} ports that we lose by
	the $p-1$ gluing operations.  Consider the gluing operation
	$G_i = G_{i-1} \oplus P_i$ and let $f_{G_{i-1}}$ and $f_P$, respectively, be
	the gluing faces identified to face $f$ of $G_i$.
	
	Observe that we counted three ports at $f_P$ since neither~$v_{n_i}$ nor
	one of its neighbors needs to assign a port to another vertex
	(we assign only ports to vertices coming later in the stacking order
	except for bend companions, but the last three vertices cannot be
	bend companions).
	We assume to lose all of these three ports when gluing.
	This means that every vertex has at most as many
	counted ports in~$G_i$ as it had in~$G_{i-1}$.
        For the ports lost at~$f_{G_{i-1}}$,
        observe that the vertex that is identified with
	$v_{n_i}$ at $P_i$ cannot lose any ports. The other two vertices are
	neighbors in $G_{i-1}$.  In the assignment that we established in
	\cref{clm:outerpathSec}, any two such vertices provide
	ports for at most four vertices in total. We assume also to lose all of
	these ports, which results in a total loss of at most seven
        ports per gluing operation.  Hence, with $p \le (n+2)/5$, we get
	\[
	\seg(G) = \frac{\openseg(G)}{2} \geq
	\frac{\openseg(\cP)-\text{loss}}{2} \geq \frac{n + 4p - 3 - (7p -
		7)}{2}
		\geq \frac{n+7}{5}\,.
	\]

	\begin{figure}[tb]
		\centering
		\begin{subfigure}[t]{.45 \linewidth}
			\centering \includegraphics[page=1]{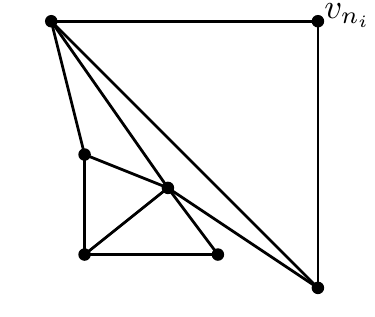}
			\caption{Initial outerplanar drawing.}
			\label{fig:TwoTreesByOuterpathsA}
		\end{subfigure}
		\hfill
		\begin{subfigure}[t]{.45 \linewidth}
			\centering \includegraphics[page=2]{seg2TreesByOuterpaths}
			\caption{Drawing after flipping $v_{n_i}$.}
			\label{fig:Two2TreesByOuterpathsB}
		\end{subfigure}
		
		\caption{Straight-line drawing of a maximal outerpath where
			we ``flip'' $v_{n_i}$ over the rest of the drawing such that
			the resulting drawing remains planar.
			This way, we can append maximal outerpaths to inner faces of 2-trees.}
		\label{fig:Two2TreesByOuterpaths}
	\end{figure}
	
	It remains to consider the case that~$G$ is a 2-tree.
	As for maximal outerplanar graphs, we can also construct a 2-tree
	by gluing multiple outerpaths.
	Similar to leaves in the dual drawing, each attached outerpath
	provides at its ending a vertex of degree~2 with two ports.
	The only exception is that we are not restricted on gluing to the
	outside~-- we may also draw a outerpath within an inner face of
	the current 2-tree drawing.
	
	A difficulty is how to identify the faces $f_{G_{i-1}}$ and $f_P$
	if we want to draw the rest of $P$ within this unified face.
	However, consider an outerplanar straight-line drawing of $P$
	where we ``flip'' the last vertex~$v_{n_i}$ over the rest of the
	drawing such that the drawing remains planar; see \cref{fig:Two2TreesByOuterpaths}.
	Clearly, the number of ports in the drawing of the maximal
	outerpath~$P$ did not change and the assignment scheme
	from \cref{clm:outerpathSec} is still applicable.
	We may use such flips also along inner edges of a outerpath
	drawing to obtain a ``folded'' outerpath drawing with the same properties.
	Hence, we can apply gluing operations to inner faces
	with at most the same loss as analyzed before.
\end{proof}

We remark that, though we get the same lower bound for
maximal outerplanar graphs and 2-trees,
the actual (tight) numbers might be different.
In other words, maybe there are 2-trees requiring less
segments than any maximal outerplanar graph with the same number of vertices.
This is because our current analysis is most likely not tight
as we see by comparison with our existential upper bound.

For an existential upper bound of maximal outerplanar graphs,
consider the construction in \cref{fig:goodouterplanar}.
\begin{figure}[tb]
  \centering
  \includegraphics{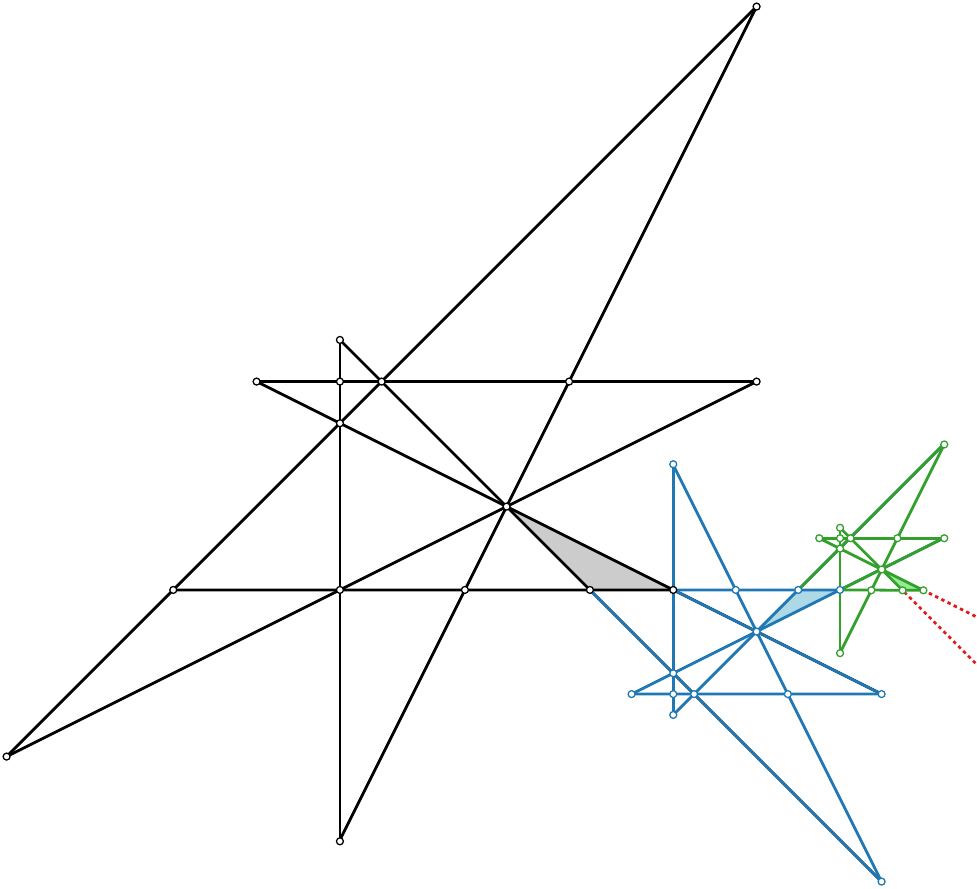}
  \caption{The maximal outerplanar graph $G_3$ with $42$ vertices
    drawn on $18$ segments. ($G_1$ in black)}
  \label{fig:goodouterplanar}
\end{figure}
It defines a family of graphs $G_1, G_2, \ldots$ where the base graph $G_1$
has 16 vertices and admits a drawing $\Gamma_{G_1}$ with eight segments.
From $G_{i-1}$ to $G_i$,
we glue a scaled and rotated copy of~$\Gamma_{G_1}$ to the drawing of~$G_{i-1}$
(gluing faces are shaded).  With each
step, we get 13 more vertices with only 5 more segments and hence the
following result.

\begin{proposition}
	\label{clm:outerplanarSecUpper}
	For every $k \in \mathbb{N}$, $G_k$ has $n_k = 13k + 3$ vertices and
	$\seg(G_k) \le 5k+3 = (5n_k+24)/13$.
\end{proposition}

\section{Planar 3-Trees}
\label{app:threetrees}

In this section we study the segment number of planar 3-trees.
For a $3$-tree $G$ with $n \geq 6$ and an arbitrary planar
straight-line drawing~$\Gamma$ of~$G$, we observe that we can assign at
least (i)~one port to each internal face of~$\Gamma$ and (ii)~twelve
ports to the outer face of~$\Gamma$; see \cref{fig:3treesegs}.  By
Euler, any $n$-vertex triangulation has $2n-5$ internal faces. Hence,
$\Gamma$ has $2n+7$ ports.  This yields the following bound, which is
tight up to a constant.

\begin{figure}[b]
	\centering \includegraphics{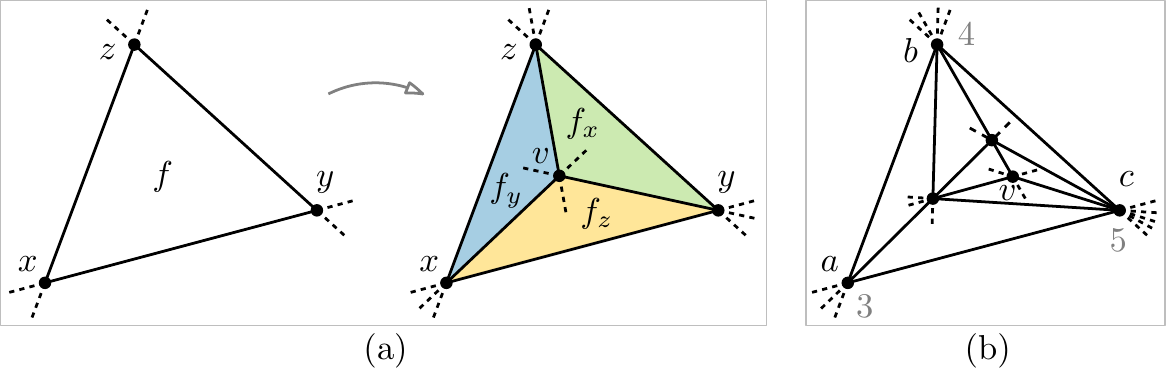}
	\caption{(a) Stacking a vertex $v$ into an internal face $f=xyz$
		creates a port in each new face ($f_x$, $f_y$, and $f_z$); (b) a
		planar 3-tree with $n \ge 6$ vertices has at least twelve ports
		on the outer face~$abc$.}
	\label{fig:3treesegs}
\end{figure}

\label{clm:3treeseg*}
\threetreeseg*

\begin{proof}%
  For claim~(i), consider a sequence of stacking operations that
  starts with a drawing of~$K_4$ and yields~$\Gamma$.  Let $v$ be the
  current vertex in this process, and let $f$ be the face into which
  $v$ is stacked.  Let $V(f)=\{x,y,z\}$ be the set of vertices
  incident to $f$, and let $f_x$, $f_y$, and $f_z$ be the three newly
  created faces such that $V(f_x)=\{v,y,z\}$ etc.; see
  \cref{fig:3treesegs}a.  Since $f$ is a triangle, no two of the
  edges $xv$, $yv$, $zv$ can share a segment.  Thus, $v$ has three
  ports.  In particular, the segment $xv$ points into $f_x$, $yv$
  points into $f_y$, and $zv$ points into $f_z$.  We assign the ports
  of $v$ accordingly to $f_x$, $f_y$, and $f_z$.  When the stacking
  process ends with~$\Gamma$, each internal face of~$\Gamma$ has a port
  assigned to it.
 
  For claim~(ii), note that the number of ports on the outer face
  equals the degree sum of the three vertices on the outer face.
  Thus, $K_4$ has nine ports.  The next (fifth) vertex in the stacking
  sequence is incident to two vertices on the outer face and hence
  contributes two more ports.  Similarly, the sixth vertex contributes
  at least one more port; see \cref{fig:3treesegs}b.  Hence, in total,
  the outer face has at least twelve ports.  (Note that this bound is
  tight since any further vertex can be stacked into an internal face
  that is not adjacent to the outer face.)
 
  To finish the proof, we treat the remaining small graphs.  For
  $n = 5$, we have one port less on the outer face, and there exists a
  drawing of this unique graph using eight segments (see
  \cref{fig:3treesegs}b without vertex~$v$).  It is easy to verify
  the claim for $n = 4$.
\end{proof}

In \cref{fig:tight3tree3}, we draw an $n$-vertex planar 3-tree
using $n + 7$ segments.
This yields an existential upper bound as formalized in \cref{clm:good3tree3}.
Hence, the universal lower bound in \cref{clm:3treeseg}
is tight up to an additive constant of~$3$.

\begin{figure}[tb]
  \centering
  \includegraphics{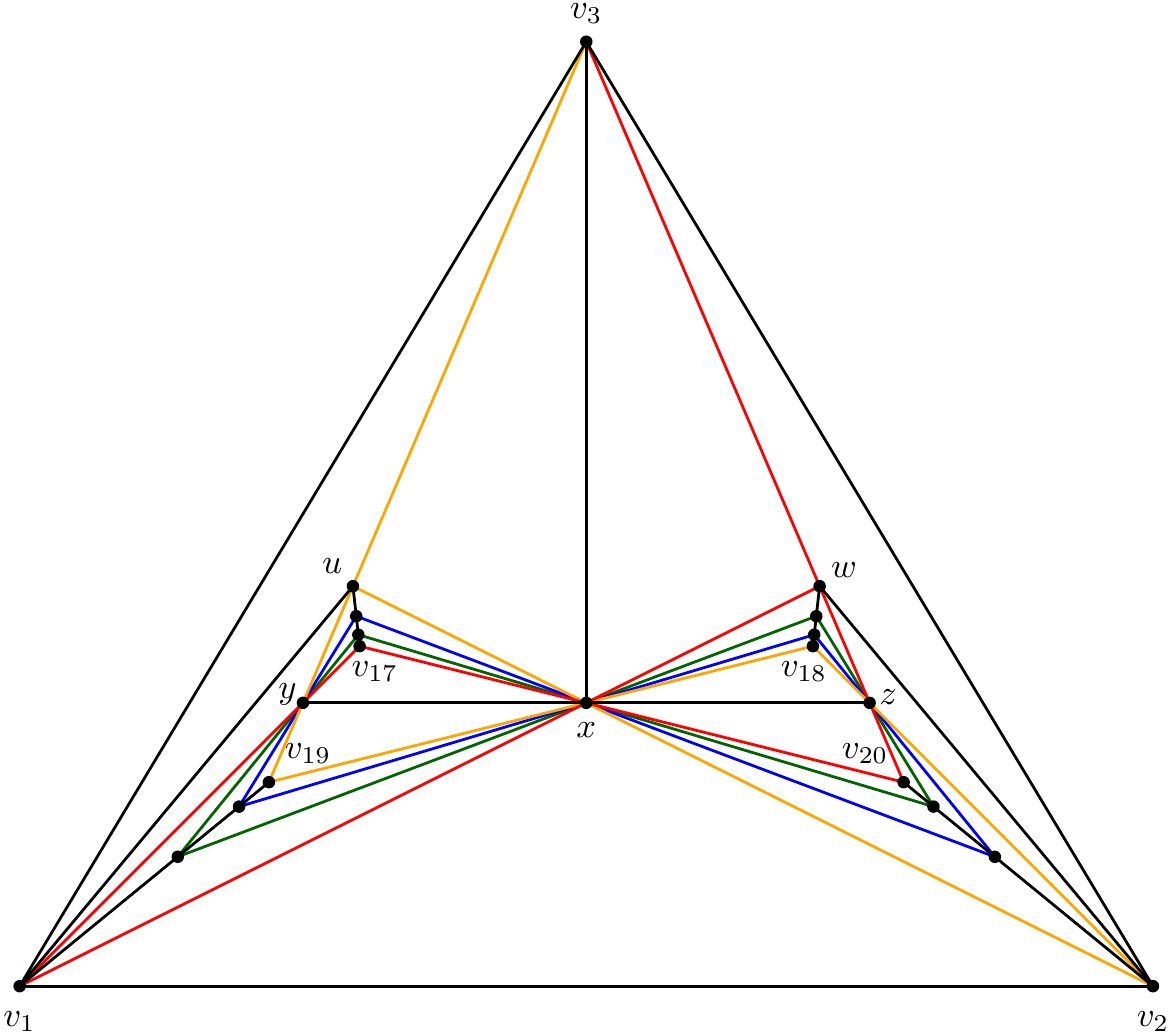}
  \caption{A straight-line drawing of $T_3$ from \cref{clm:good3tree3}
    with $20$ vertices and $27$ segments.}
  \label{fig:tight3tree3}
\end{figure}

\begin{proposition}
	\label{clm:good3tree3}
	For every $k \ge 1$ there exists a 3-tree~$T_k$, whose
	construction is illustrated in \cref{fig:tight3tree3}, with
	$n = 4k + 8$ vertices and $\seg(T_k) \le 4k+15 = n+7$.
\end{proposition}

\begin{proof}
	Consider \cref{fig:tight3tree3}.
	We start by drawing the outer triangle $\triangle v_1v_2v_3$
        using three segments.
	As fourth vertex, we add the central vertex~$x$ introducing three more segments.
	For the fifth and sixth vertex, $u$ and $w$, we re-use
	the line segments $xv_1$ and $xv_2$ and, consequently,
	add only four new segments.
	For the seventh and eighth vertex, $y$ and $z$, we re-use
	line segments~$uv_3$ and~$wv_3$, respectively.
	Moreover, they share a segment for the edges $yx$ and $zx$,
	which results in three new segments.
	This gives us 13 vertices for the base construction.
	
	Now in $k$ rounds, we iteratively stack four vertices into the
        faces $\triangle uyx$, $\triangle wxz$, $\triangle v_1xy$, and
        $\triangle v_2zx$.        
	We stack along four new (black) line segments (see e.g. $\overline{v_1 v_{19}}$
	in \cref{fig:tight3tree3}) such that the final drawing uses four more segments
	once as well as four more per iteration
	(colored line segments through $y$, $z$ and $x$ in \cref{fig:tight3tree3}).
	We re-use the segments $uy$, $wz$, $v_1y$, and $v_2z$
	for one edge each, which saves us two more segments.
	Together with the 13 segments of the base construction, we get
	$\seg(T_k) \le 13 + 4 + 4k - 2 = 4k + 15 = n + 7$.
\end{proof}

Consider the universal upper bound of $2n - 2$ on the segment number
of planar 3-trees due to Dujmovi\'c et
al.~\cite[Lemma~18]{desw-dpgfss-CGTA07}.  They show the tightness of
their result in a fixed-embedding setting, that is, they prove that
there is a family $(B_n)_{n \ge 4}$ of {\em plane} 3-trees (see
\cref{fig:3treestackedbad}) such that $B_n$ has $n$ vertices and
requires $2n-2$ segments in any straight-line drawing that adheres
to the given embedding.  They remark that, given a different
embedding, $B_n$ can be drawn using roughly $3n/2$ segments; see
\cref{fig:3treestackedgood}.  We formalize this to compute the exact
segment number of~$B_n$, which will be useful in
\cref{app:discussion}.

\begin{figure}[tb]
	\begin{subfigure}[t]{.31 \linewidth}
		\includegraphics[page=1]{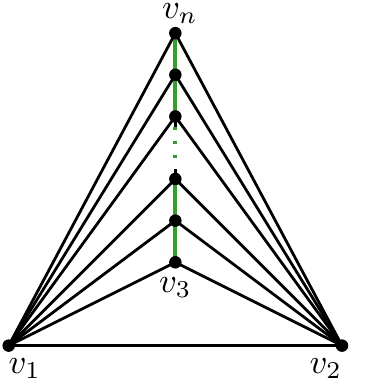}
		\caption{$2n - 2$ segments}
		\label{fig:3treestackedbad}
	\end{subfigure}
	\hfill
	\begin{subfigure}[t]{.31 \linewidth}
		\includegraphics[page=2]{seg3TreeStackedTriangles}
		\caption{$3n/2 + 1$ segments}
		\label{fig:3treestackedgood}
	\end{subfigure}
	\hfill
	\begin{subfigure}[t]{.31 \linewidth}
		\includegraphics[page=3]{seg3TreeStackedTriangles}
		\caption{$2n - 2$ segments}
		\label{fig:3treestackedbaddeg3}
	\end{subfigure}
	\centering
	\caption{Straight-line drawings of the 3-tree $B_n$
          (with $n \ge 6$ even) for two
          different embeddings that were analyzed by Dujmovi\'c et
          al.~\cite{desw-dpgfss-CGTA07}}
	\label{fig:3treestackedtriangles}
\end{figure}

\begin{proposition}
  \label{clm:bad3tree3}
  For every $n \ge 6$ there exists a 3-tree~$B_n$ (see
  \cref{fig:3treestackedtriangles}) with $n$ vertices and
  $\seg(B_n) = \lceil 3n/2 \rceil + 1$.
\end{proposition}

\begin{proof}
  We first show the lower bound $\seg(B_n) \ge \lceil 3n/2 \rceil + 1$.
	
  Let $B_n$ be the graph depicted in \cref{fig:3treestackedtriangles}
  with vertex set $\{v_1,v_2,\dots,v_n\}$ and edge set $\{v_1v_i,v_2v_i
  \colon 3 \le i \le n\} \cup \{v_i v_{i+1} \colon 1 \le i \le n-1\}$.
  If $v_1$ and $v_2$ are on the outer face (see \cref{fig:3treestackedbad}),
  we have at least $2 (n-2) + 2$ segments.  For $n \ge 6$,
  $2n-2 \ge \lceil 3n/2 \rceil + 1$.
	
	So w.l.o.g. let $v_1$ not be on the outer face.
	Consequently, $v_2$ lies on the outer face because any triangle
	of the graph contains at least one vertex of $\{v_1, v_2\}$ and, hence,
	also the triangle of the outer face.
	This implies that there are $n-1$ distinct segments incident to $v_2$.
	
	For every $i \in \{3, \dots, n\}$, the path $\langle v_1, v_i,
        v_2 \rangle$ is drawn
	with a bend at $v_i$ because otherwise it would coincide with the
	edge $v_1v_2$.
	Therefore, the $n-2$ edges $v_1v_3,\dots,v_1v_n$ form at least
        $(n-2)/2$ new segments.
	
	Consider the two other vertices on the outer face~-- we call them $u$ and $w$.
	The edge $uw$ yields another segment.
	Moreover, $v_3$ and $v_n$ cannot both be on the outer face
	as they are not adjacent.
	Therefore, w.l.o.g., $u$ has degree 4.
	So far, we have counted the segments of the edges
	$uv_1$, $uv_2$ and $uw$.
	This means that there is another segment for the fourth
        edge incident to~$u$.
	
	If $w$, too, has degree 4,
	we count another segment by the same argument.
	Overall, this sums up to at least
        $(n-1) + (n-2)/2 + 1 + 2 = 3n/2 + 1$ segments.	
	
	Otherwise $w$ has degree 3.
	Assume w.l.o.g.\ that $w = v_3$.
	Consequently, the outer face is the triangle
	$\triangle v_2v_3v_4$; see \cref{fig:3treestackedbaddeg3}.
	Observe now that $\triangle v_1v_2v_4$
	separates $v_3$ on the outside from
	all other vertices in the inside.
	Thus, the $n-3$ edges $v_1v_4,\dots,v_1v_n$
	reach $v_1$ in an angle smaller than $180^\circ$
	and, hence, require $n-3$ distinct segments.
	This results in at least
        $(n-1) + (n-3) + 1 + 1 = 2n-2 \ge 3n/2 + 1$ segments.
	
	Finally, we show that this lower bound is tight.
	Consider the drawing of $B_n$ in \cref{fig:3treestackedgood}.
	It uses exactly the $\lceil 3n/2 \rceil + 1$ segments that we counted above
        for the lower bound.
	In particular, $u = v_{\lfloor n/2 \rfloor+1}$ and $w = v_{\lfloor n/2 \rfloor+2}$.
\end{proof}

\section{The Ratio of Segment Number and Arc Number}
\label{app:discussion}

Since circular-arc drawings are a natural generalization of
straight-line drawings, it is natural to also ask about the maximum
ratio between the segment number and the arc number of a graph.  In
this section, we make some initial observations regarding this
question. Clearly, $\seg(G) / \arc(G) \ge 1$ for any graph~$G$.
Note that $\seg(K_3)/\arc(K_3)=3$.  We investigate the ratio
for two classes of planar graphs.  We construct families of graphs
showing that, for maximal outerpaths, (and, hence, for maximal
outerplanar graphs and 2-trees) the minimum ratio is~1
(\cref{clm:2treeMinRatio}/\cref{fig:tightOuterpaths}a)
and the maximum ratio is at least~2
(\cref{clm:2treeMaxRatio}/\cref{fig:tightOuterpaths}b).
For planar 3-trees, the minimum ratio is at most~$4/3$
(\cref{clm:3treeMinRatio}/\cref{fig:tight3tree3}) and the maximum
ratio is at least~3
(\cref{clm:3treeMaxRatio}/\cref{fig:arc3treexmastree}).

It would be interesting to find out how much of an improvement in
terms of visual complexity circular-arc drawings offer over
straight-line drawings for arbitrary planar graphs.  Can the ratio
between segment and arc number be bounded by~3 for every planar graph?

\begin{proposition}
  \label{clm:2treeMinRatio}
  For $r \in \mathbb{N}$, let $P_r$ be the maximal outerpath from
  \cref{prop:outerpathExamples}; see \cref{fig:tightOuterpaths}a.
  Then, $\lim_{r\to\infty} \seg(P_r)/\arc(P_r)=1$.
\end{proposition}

\begin{proof}
  Consider \cref{fig:tightOuterpaths}a for a drawing of $P_r$ on
  $n/2+2$ segments where $n$ is the number of vertices of~$P_r$.
  Observe that the central vertex has degree~$(n-1)$ and, thus,
  is contained in at least $n/2$ different arcs in any arc-drawing.  Hence
  the segment number and the arc number of~$P_n$ differ by at
  most a constant of~$2$.
\end{proof}

\begin{proposition} \label{clm:2treeMaxRatio}
  For every positive integer $k$, let $Q_k$ be the maximal outerpath
  with $n_k=3k+3$ vertices shown in \cref{fig:tightOuterpaths}b.
  Then $\lim_{k\to\infty} \seg(Q_k)/\arc(Q_k) \ge 2$.
\end{proposition}

\begin{proof}
  The outerpath $Q_k$ contains $n_k/3 - 2$ degree-$6$ vertices and for
  each of them two degree-$3$ neighbors, with at least one port each.
  The degree-$6$ vertices either have two ports themselves
  or their bend companions have three ports.  In either case, we find
  four ports for each degree-$6$ vertex.
  The remaining six vertices around the first and the last face have at
  least ten ports.  Therefore, $\seg(Q_k) \geq 2n_k/3 + 1$.  
  \cref{fig:tightOuterpaths}b yields that $\arc(Q_k)\leq n_k/3+1$.
  Hence, $\seg(Q_k)/\arc(Q_k)\geq 2-2/(k+2)$.
\end{proof}

\begin{proposition}
  \label{clm:3treeMinRatio}
  For $k\ge2$, let $T_k$ be the planar $3$-tree shown in
  \cref{fig:tight3tree3}.  Then
  $\lim_{k\to\infty} \seg(T_k)/\arc(T_k) \le 4/3$.
\end{proposition}

\begin{proof}
  See \cref{fig:tight3tree3} for a drawing of $T_k$ on $4k+11$
  segments.  Let $v$ be the unique vertex of degree $n-1$ and $u$, $w$
  be the two degree-$(n/2-1)$ vertices.

  There is a set of $4k+2$ unique paths, one half from $u$ to $v$ and
  the other from $v$ to $w$.  Each of these paths needs to be covered
  by at least one arc.  Obviously no arc may cover more than two
  paths.  Now observe that any arc covering one path on each side
  connects all three vertices, such that only one such arc may
  exist. Hence, of the remaining $4k$ paths we may cover only two with
  the same arc if both lie on the same side of~$v$.  However, every
  such arc must have the same tangent in $v$ in order not to cross
  the other paths.  Therefore, we may only do this on one side of~$v$.
  Thus, $k$ arcs may suffice for one side, but the other needs
  $2k$ arcs, which yields a total of $3k+1$ necessary arcs.  Hence,
  $\seg(T_k)/\arc(T_k) \leq 4/3+29/(9k+3)$.
\end{proof}

\begin{proposition}
  \label{clm:3treeMaxRatio}
  For every even $n\ge 8$, let $B'_n=B_n - v_1v_2$ be the planar graph
  shown in \cref{fig:arc3treexmastree}.  Then
  $\seg(B'_n)/\arc(B'_n) = 3$ and
  $\lim_{n\to\infty} \seg(B_n)/\arc(B_n) = 3$.
\end{proposition}

\begin{figure}[tb]
  \includegraphics[page=1]{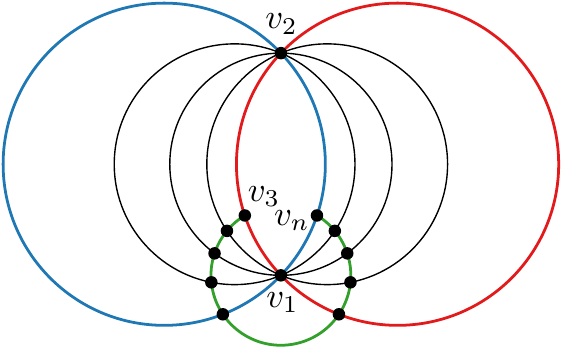}
  \hfill
  \includegraphics[page=2]{arc3TreeXmasTree}
  \caption{The planar graph $B'_{12}$ from \cref{clm:3treeMaxRatio}
    drawn with $n/2=6$ arcs and with $3n/2=18$ line segments.}
  \label{fig:arc3treexmastree}
\end{figure}

\begin{proof}
  \Cref{fig:arc3treexmastree} shows drawings of $B'_n$ with $n/2$ arcs
  and with $3n/2$ segments.  Clearly, $\arc(B'_n)=n/2$ since
  $\deg(v_2)=n-2$ and there are two vertices of odd degree ($v_3$
  and $v_n$), where some arc(s) must start and end.  For the same
  reason $\arc(B_n)=n/2+1$.  By \cref{clm:bad3tree3},
  $\seg(B_n)=3n/2+1$.  Recall that removing the edge~$v_1v_2$
  from~$B_n$ yields~$B'_n$.
  Observe that $B'_n$ is still triconnected.
  Therefore, the set of embeddings is the same as for $B_n$
  (except that we have the face $\langle v_1, v_n, v_2, v_3 \rangle$
  instead of the triangular faces $\triangle v_1v_2v_3$
  and $\triangle v_1v_nv_2$)
  and depends only on the choice of the outer face.
  Analyzing the different embeddings
  of~$B'_n$ as those of~$B_n$ in the proof of \cref{clm:bad3tree3},
  shows that $\seg(B'_n)=3n/2$.  In particular, while we could
  straighten the path $\langle v_2,v_3,v_1\rangle$ in
  \cref{fig:arc3treexmastree}, this would introduce a new bend in the
  path $\langle v_3, v_1, v_{n/2+2}\rangle$, and the number of
  segments remains $3n/2$.  Hence $\seg(B'_n)/\arc(B'_n)=3$ and
  $\lim_{n\to\infty} \seg(B_n)/\arc(B_n) = 3$.
\end{proof}

\end{document}